\documentclass[11pt,a4paper]{article} 
\usepackage[english]{babel}
\usepackage[margin=1in,footskip=0.25in]{geometry}
\usepackage{graphicx}%
\usepackage{multirow}%
\usepackage{amsmath,amssymb,amsfonts}%
\usepackage{amsthm}%
\usepackage{mathrsfs}%
\usepackage[title]{appendix}%
\usepackage{xcolor}%
\usepackage{textcomp}%
\usepackage{manyfoot}%
\usepackage{booktabs}%
\usepackage[ruled,vlined]{algorithm2e}
\usepackage{listings}%
\usepackage{hyperref}
\usepackage{subcaption}
\usepackage[T1]{fontenc}
\usepackage[utf8]{inputenc}

\theoremstyle{thmstyleone}%
\newtheorem{theorem}{Theorem}
\newtheorem{proposition}[theorem]{Proposition}%

\theoremstyle{thmstyletwo}%
\newtheorem{remark}{Remark}%

\theoremstyle{thmstylethree}%
\newtheorem{definition}{Definition}%
\newtheorem*{hypothesis}{Hypothesis}%
\newtheorem{lemma}{Lemma}%
\newcommand{\E}{\mathbb{E}}
\newcommand{\F}{\mathbb{F}}
\newcommand{\G}{\mathbb{G}}
\newcommand{\p}{\mathbb{P}}
\newcommand{\R}{\mathbb{R}}
\newcommand{\ind}{\mathbf{1}}
\newcommand{\T}{T \wedge \tau}

\DeclareMathOperator*{\esssup}{ess\,sup}
\title{Delegated portfolio management with random default}
\author{Alberto {\sc GENNARO}\footnote{ Department of Industrial Engineering and Operations Research, UC Berkeley, USA.\\ alberto.gennaro@berkeley.edu}~~and Thibaut {\sc MASTROLIA} \footnote{ Department of Industrial Engineering and Operations Research, UC Berkeley, USA.\\ mastrolia@berkeley.edu},}

\begin{document}

\maketitle
\begin{abstract}

We consider the problem of optimal portfolio delegation between an investor and a portfolio manager under a random default time. We focus on a novel variation of the Principal-Agent problem adapted to this framework. We address the challenge of an uncertain investment horizon caused by an exogenous random default time, after which neither the agent nor the principal can access the market. This uncertainty introduces significant complexities in analyzing the problem, requiring distinct mathematical approaches for two cases: when the random default time falls within the initial time frame $[0, T]$ and when it extends beyond this period. We develop a theoretical framework to model the stochastic dynamics of the investment process, incorporating the random default time. We then analyze the portfolio manager's investment decisions and compensation mechanisms for both scenarios. In the first case, where the default time could be unbounded, we apply results from Backward Stochastic Differential Equations (BSDEs) and control theory to address the agent's problem. In the second case, where the default time is within the interval $[0, T]$, the problem becomes more intricate due to the degeneracy of the BSDE's driver. For both scenarios, we demonstrate that the contracting problem can be resolved by examining the existence of solutions to integro-partial Hamilton-Jacobi-Bellman (HJB) equations. We develop a deep-learning algorithm to solve the problem with no access to the optimizer of the Hamiltonian function by means of an actor--critic algorithm.  
\end{abstract}
\textbf{Keywords: }{Stochastic control with random horizon, Principal-Agent problem, enlargement of filtration, BSDE, HJB equation and deep learning.}

\section{Introduction}\label{sec1}

Delegating portfolio management from an investor to a professional fund manager is increasingly seen as a strategic move due to the growing complexity of financial markets and their fragmentation \cite{jensen1968performance,bernstein1998investment,stracca2006delegated}. The financial landscape today is marked by rapid market fluctuations, changing regulations, and an extensive range of investment opportunities, all of which necessitate not only significant time and effort but also in-depth knowledge and experience to navigate effectively. For many investors, handing over portfolio management to a professional allows them to leverage the fund manager’s expertise in areas such as professional oversight, diversification strategies, and sophisticated risk management techniques, capabilities that are often difficult to achieve independently.\\

In this context, fund managers are expected to deliver superior performance by utilizing their specialized skills and tools. Traditionally, the compensation structure for fund managers has been a blend of a fixed fee and a performance-based component. The fixed fee provides a stable income for the manager, while the performance-based component is designed to incentivize the manager to achieve better returns by aligning their interests with those of the investor. However, this conventional compensation structure raises an important question: Is it optimally designed to align the incentives of both the investor and the fund manager? The challenge lies in ensuring that the performance-based component effectively motivates the fund manager to act in the best interests of the investor, while also accounting for the inherent uncertainties and risks associated with market fluctuations.
To address this, it is crucial to evaluate whether the existing compensation structures are adequately aligned with the investor’s objectives and whether alternative models could offer better alignment. This involves exploring various compensation schemes and their impact on fund performance, risk management, and overall investor satisfaction. Ultimately, a well-designed compensation structure should incentivize fund managers not only to maximize returns but also to manage risks prudently, ensuring that both the investor’s and the manager’s goals are harmoniously aligned in the pursuit of financial success.\\

From a mathematical perspective, this issue can be framed as a variation of the principal-agent (PA) problem in wealth management, see, for example, \cite{dalmacio2004agency,ou2003optimal,stracca2006delegated,demarzo2006optimal,li2009incentive,leung2014continuous,cvitanic2017moral}. The PA framework in continuous time is a game-theoretical model designed to address problems of stochastic control, where one party (the principal) delegates decision-making authority to another party (the agent), whose actions are not directly observable and evolve in a stochastic environment. The agent controls the system by choosing a strategy that influences outcomes, but because of the information asymmetry, the principal cannot directly observe the agent's actions. Instead, the principal must design a contract based on the observed outcomes, which indirectly depend on the agent's actions. The goal is to structure this contract in a way that motivates the agent to exert optimal effort, while managing the inherent trade-offs between risk-sharing and incentives. Here, the investor (the principal) has an initial capital, $X_0 = x$, and seeks an agent to invest on their behalf. The principal is willing to negotiate a compensation scheme that incentivizes the agent based on portfolio performance and risks. In the literature, the PA problem was extensively studied in a simpler setting: in the seminal work of \cite{holmstrom1987aggregation}, the agent controls the drift of the process, but the utility is just drawn from the terminal value of the controlled process. In \cite{sannikov2008continuous} we still see a control based on just the drift, but the utility is now directly dependent on inter-temporal payments. In both cases, because of the modeling choice of having a single Brownian motion, there is no moral hazard with respect to the control on the volatility. A concept introduced by \cite{cvitanic2017moral}, moral hazard in the volatility control is important when multiple sources of risks arise in the control problem. What the authors proved in this work, thanks to the mathematical advancement on singular change of measure, is that in optimal contracts there are some incentives to be given with respect to the quadratic variation of the controlled process and the co-variation with respect to the risk factors. Furthermore, in their work, instead of working from a probability perspective, they cast the problem as a stochastic control problem, and they restrict the analysis to an admissible family of contracts, proving no loss of generality in doing so, paving a simpler way to the recent literature on the subject. In the subsequent paper \cite{cvitanic2018dynamic}, the same result was proven by using second-order BSDE, and a recent work \cite{chiusolo2024new} simplified even more the theoretical guarantees for volatility controlled problems, showing the same results but using the standard theory of BSDE. To summarize, in these works, they define a class of admissible contracts, they prove there is no loss of generality in considering such a form, and they find the optimal contract in this set. In practice, they rely on the work of \cite{pardoux1990adapted,barles1997backward,el1997backward,pardoux2005backward,kobylanski2000backward,papapantoleon2018existence} on the existence of solutions for BSDEs in the Lipschitz and quadratic cases with or without jumping terms and random horizon. This will also be the building block for this work.\\

Portfolio optimization has been extensively studied in the literature, beginning with the pioneering work of \cite{markowitz1991foundations}. More recent studies have explored continuous-time versions of the framework, as well as settings that incorporate jumps, as in \cite{morlais2009utility}. The PA problem captures the scenario of a fund manager investing on behalf of an investor. Despite its practical relevance, the existing literature has not fully addressed a crucial aspect of this problem: the randomness of the investment horizon. More often than not, an investment in financial markets does not have a precise duration, and duration is used as a reference for setting risk aversion for the investor.

Our contribution aims to fill this gap by examining how decision-making strategies for both the agent's investment and the principal's compensation scheme are affected by the introduction of default times, adding uncertainty to the investment horizon. Adding a default time makes the problem mathematically more challenging. First of all, using a general default time forces one to delve into the depths of information theory. To be able to treat the problem, we will need to enlarge the filtration, adding to the one generated by the financial market the filtration resulting from the random default (see \cite{jeanblanc2009mathematical,bielecki2013credit,aksamit2017enlargement,Guo_2008} and the references therein). Furthermore, as highlighted by \cite{jeanblanc2015utility}, there are two distinct cases to consider for default times, each with different implications:
\begin{itemize}
    \item \textbf{Unbounded Case}: If the maximum possible default time $S$ exceeds the investment horizon $T$ (or is infinite), it is uncertain whether a default will occur within the investment period. In this case, the investing problem is reduced to a utility maximization under random horizon. It has been solved in, for example, \cite{ kharroubi2013mean} which proves that the solution is related to a system of BSDEs with a jump admitting a solution via a decomposition approach coming from filtration enlargement theory. Without loss of generality, we can assume that $S = +\infty$, hence the name unbounded.
    \item \textbf{Bounded Case}: If $S$ is less than $T$, default will certainly occur before the investment horizon ends, but the exact timing is unknown. Without loss of generality, we can assume $S = T$: if the default happens almost surely before $T$, then we can simply rewrite the a priori horizon as $S$. It has been solved in, for example, \cite{jeanblanc2014note,jeanblanc2015utility} proving that the solution is related to a system of \textit{degenerate} BSDEs with a jump.
\end{itemize}
These two cases have not only different interpretations but also different mathematical tools required to solve them. 
We motivate our study with two contrasting examples inspired by
cryptocurrency trading and life–insurance portfolio management.
In highly speculative markets such as crypto assets, the investment
horizon is intrinsically uncertain: markets may collapse abruptly due to technological failures, regulatory bans, or market inefficiencies that effectively ``terminate'' the investment opportunity. Therefore, we are in a situation where the intended investment horizon can be achieved, but a black-swan event may cause the market to cease to exist.
The situation is completely different for life–insurance companies. Although the investment horizon of most products is in the order of $100$ years, the effective horizon remains random, as the contract expires at the death time of the insured. Furthermore, regulatory and actuarial considerations impose strict constraints on the risk profile of managed portfolios. Insurers are required to invest in a manner that does not deviate too aggressively from prescribed benchmark allocations, such as long-duration bond indices or liability-driven reference portfolios.  
A convenient way to encode such tracking discipline is to introduce a non–negative parameter $\varepsilon \ge 0$ that penalizes deviations between the insurer's strategy and some benchmark, for instance through a running quadratic tracking error.  
In this interpretation, the crypto setting corresponds to $\varepsilon=0$,
whereas the life-insurance application is characterized by $\varepsilon>0$, thus bridging two seemingly different investment problems within a unified principal-agent framework under random horizon.

In general, therefore, the unbounded case can be seen as a default caused by a black swan event \cite{pate2012black}, or a crash that forces authorities to close the market (the Flash crash in May 2010, see \cite{kirilenko2017flash}), or, in the blockchain, a hacker attack or, finally, in a more structured deal with investors, when a fund sees its money withdrawn with no or little notice. This often complicates the investment strategies for funds, and for this reason certain funds (i.e. hedge funds) have very strict policies on funding withdrawal. Mathematically, the BSDE related to this case is better behaved than the other case. The bounded case instead can be representative of the well-known and studied life-insurance market: in this case, the insurance policy can have a time horizon of over $100$ years, so that we can claim, with probability one, that the investor is going to pass away before the natural termination of the contract. This means that, when calibrating the contract, both the agent and the investor are aware of the fact that the horizon will not be reached, and this is taken into account in both the trading strategies and the insurance payments. Mathematically, this formulation introduces a difficulty in the family of proposed contracts, as it generates a BSDE with a singular driver (see \cite{jeanblanc2015utility} and references therein), and it also poses some difficulties in the convergence of numerical methods.\\

The problem's structure also presents challenges for numerical solutions. The partial differential equation (PDE) resulting from the Hamilton-Jacobi-Bellman (HJB) control problem has a varying coefficient dependent on the solution of a maximization problem involving the solution itself. Addressing this requires a specialized approach using an ``actor--critic''
 iterative algorithm, where the actor solves the PDE for a fixed coefficient, and the critic updates the maximization problem based on the actor's latest guess. While several schemes could tackle this iterative process, the most effective have proven to be in the domain of Physics-Informed Neural Networks (PINNs). PINNs are a powerful machine learning method that blends neural networks with principles from physics to solve complex differential equations, especially in situations where traditional methods may struggle. From a technical perspective, the surge of this methodology was possible because of one of the most useful but perhaps underused techniques in
scientific computing, automatic differentiation. We refer to the survey \cite{baydin2018automatic} for a comprehensive study of this problem. The simple idea behind it is to differentiate neural networks with respect to both their input coordinates and model parameters: the former allows us to have derivatives in the loss function, the latter is the standard way to train the network. Introduced and expanded by works such as \cite{raissi2017physics} and \cite{sirignano2018dgm}, PINNs leverage the underlying physical laws, typically encoded as partial differential equations (PDEs), to guide the learning process. Rather than relying purely on data, PINNs incorporate these governing equations into the loss function, ensuring that the neural network solutions respect known physical constraints. This approach is particularly effective for solving high-dimensional partial (integro) differential equations, arising from (stochastic and continuous-time) problems in virtually every field, like fluid dynamics, electromagnetism, biology, or finance (see the work of \cite{baldacci2019market}, \cite{baldacci2022governmental}). By incorporating physics directly into the architecture, PINNs enable the modeling of complex systems while reducing reliance on large datasets, bridging the gap between traditional numerical solvers and modern machine learning techniques. Despite all these difficulties, this default-time formulation is crucial for practical applications, as it makes both contract incentives and trading strategies more robust. This study contributes significantly to the literature on principal-agent problems, extending its applicability to real-world financial scenarios by providing some insights into the effects of random investment horizons and default times.\\

The organization of the paper is as follows. In Section \ref{chap2}, we present the mathematical formulation of the Principal-Agent problem under time uncertainty, describing the underlying stochastic framework, the controlled wealth process governing the system dynamics, and all the necessary assumptions to make the problem tractable. We will further define some classes of admissible contracts, and the optimization problems for both the principal and the agent. In Section \ref{chap3}, we solve the problem sequentially, first focusing on the agent’s optimal strategy, which is going to be the same for both cases of bounded and unbounded default. Despite the fundamental difference in the proof of existence, this trading strategy has the same form in both cases and will be plugged into the principal's problem. Then we will derive the Hamilton-Jacobi-Bellman (HJB) equation for the principal and we will claim, with a verification theorem, the existence of the solution to the partial differential equation (PDE) that encapsulates the principal’s optimization problem in the two proposed settings. Then, Section \ref{chap4} provides numerical examples that demonstrate the implementation of the theoretical results in concrete scenarios, for both cases, using default times from the families of the Beta and Exponential distributions. The goals are multiple: to show the differences arising within the same case, but also across the two different cases. Finally, we also want to highlight the sub-optimal behavior of a subset of optimal contracts that mimic real-world compensation schemes.\\

The code of our numerical simulations is publicly available on \url{https://github.com/GennaroAlberto/PrincipalAgent}.

\section{The model and the optimization problem}\label{chap2}

\subsection{Risky assets and portfolio dynamics}
We consider a financial market represented by a probability space $(\Omega,\mathcal F,\mathbb P)$ endowed with a $d$-dimensional Brownian motion denoted by $W$ and a finite horizon $T>0$. We denote by $\mathbb F:=(\mathcal F_t)_{t\in[0, T]}$ the natural filtration of this Brownian motion. This market consists of $m$ risky assets with vector price $S_t$ at time $t$ with no risk-free rate. The risky assets have the following dynamics:
\begin{equation*}
    \frac{dS_t^i}{S_t^i} = b^i_t dt + \sigma^i_t dW_t \qquad \forall i = 1: m, 
\end{equation*}

where $\sigma^i,b^i$ are respectively $\mathbb R^{1\times d}-$valued and  $\mathbb R$-valued bounded $\mathbb F-$predictable processes. We define the $\mathbb R^{m\times d}-$valued covariance matrix $\sigma$ where $\sigma^{i,j}$ is the $j$-th component of $\sigma^i$. We assume that $\sigma\sigma^\top$ is an invertible matrix, that is $\sigma\sigma^\top$ is $\mathbb P\otimes dt-$a.e. elliptic. Note that 
\[dS_t = B_t dt+\Sigma_t dW_t,\]
where $S=(S^1,\dots,S^m)$, $B=(b^i S^i)_{1\leq i\leq m}$ and $\Sigma$ is a $m\times d$ matrix with the $i$-th row is $\Sigma^i_t=S_t^i \sigma^i_t$.

We define $\theta_t = \sigma_t^\top (\sigma_t \sigma_t^\top)^{-1} b_t$. Let $\pi_t = (\pi_t^i)_t$ be a vector in $\mathbb R^{1\times m}$ representing the fraction of money invested in every asset at time $t$. We refer to it as the \textit{investment strategy} of the portfolio manager. We set $\beta_t = \pi_t \sigma_t$. Note that we can also refer to $\beta$ or $\pi$ as the investment strategy interchangeably up to the volatility factor $\sigma$. For every $\pi$, one can define a probability measure $\mathbb P^\pi$ such that the dynamics of the value of the portfolio  starting with $X_0 = x$ is then given by\footnote{We refer to Appendix A in \cite{baldacci2022governmental} for the rigorous formulation of the problem and the choice of the probability $\mathbb P^\pi$.}
\[
    X_t := x + \int_0^t \pi_s \sigma_s dW_s + \int_0^t \pi_s b_s ds, \]
    or equivalently
\[ X_t:= x + \int_0^t \beta_s dW_s + \int_0^t \beta_s \theta_s ds.
\]

\subsection{Default time and enlargement of filtration}
The default time is represented by a random variable $\tau$ taking values in $\mathbb R^+$. We define the default process by $H_t := \ind_{\tau \leq t}$. Note that this process is not necessarily $\mathbb F-$ measurable, and may require extra information or can be simply exogenous to the system and so independent of $W$. We thus enlarged the available information and also the filtration $\mathbb F$ taking into account the information generated by the default time occurrence. 
\begin{definition}
   Let $\mathcal H_u:=\sigma(H_s, s \in [0, u])$ be the $\sigma-$algebra generated by $H$ until time $u\geq 0$. Given a filtered space $(\Omega, \mathcal{F}_T, \F, \mathbb{P})$, the enlarged filtration
    \[
        \G = (\mathcal{G}_t)_{t \in [0,T]}, \qquad \mathcal{G}_t = \underset{\epsilon > 0}{\bigcap} \{ \mathcal{F}_{t+\epsilon} \vee \mathcal H_{t +\epsilon} \},
    \]
    is the smallest enlargement of $\F$ such that $\tau$ is a $\G$-stopping time.
\end{definition}

\begin{remark}
$H_t$ is not $\mathbb F$-measurable but it is $\G$-measurable stochastic process. 
\end{remark}
The goal is to ensure that the inaccessible default time $\tau$ enables us to enlarge the filtration to $\mathbb G$ and transferring the martingale property from $\mathbb F$ to $\mathbb G$ known as the immersion property or H-hypothesis. The first fundamental hypothesis is to set the existence of a (conditional) density for the default time with a certain property.

\begin{hypothesis}[Density Hypothesis]\label{densityhyp}
For any $t\geq 0$, there exists an $\mathcal F_t\otimes \mathcal B((0,\infty))$ measurable map $\gamma(t, \cdot)$ such that
\[   
    \p(\tau \geq x | \mathcal F_t) = \int_{x}^\infty \gamma(t, u) du \qquad \forall x \geq 0
\]
and $\gamma(t, u) = \gamma(u,u) \ind_{t\geq u}$ 
\end{hypothesis}
As a consequence of this assumption, see for example \cite{bremaud1978changes, el2010happens}, any $\F$-martingale is also a $\G$-martingale. Furthermore, and still under the density hypothesis, the process $H$ admits an absolutely continuous compensator, \textit{i.e.,} there exists a non-negative $\G$-predictable process $\lambda^{\G}$, such that the compensated process $M$ defined by
\[M_t := H_t - \int_{0}^t \lambda_s^{\G} ds
\]
is a $\G$-martingale. The compensator vanishes after time $\tau$ (therefore $\lambda^{\G}_t := \lambda_t \ind_{t \leq \tau}$) and 
\[
    \lambda_t := \frac{\gamma(t,t)}{\p(\tau > t | \mathcal F_t)}
\]
is a $\F$-predictable process. For a complete and deeper discussion on the properties of enlarged filtrations, we refer the reader to \cite{Guo_2008}. We set $\Lambda_t:=\int_0^t \lambda_s ds$. As a consequence of Proposition 4.4 in \cite{el2010happens} 
\[\mathbb P(\tau>t|\mathcal F_t)=e^{-\Lambda_t}.\]
We now turn to the integrability of the process $\lambda$ and the support of the default time $\tau$. Denoting by $\mathcal{T}(\mathbb{A})$ the set of $\mathbb{A}$-stopping times (so we will have $\mathcal{T}(\F)$ or $\mathcal{T}(\G)$), we consider two cases\\

\noindent \textit{Hypothesis A - unbounded default.} 
\begin{equation} \label{H:unbounded}
    \tag{HA}
    \esssup_{\rho \in \mathcal{T}(\mathcal G)} \E\Bigl[\int_{\rho}^T \lambda_s ds \, | \,  G_{\rho}\Bigl] < + \infty
\end{equation}

As a direct consequence of the tower property and since $\F \subseteq \G$, Hypothesis A leads to
\[
    \esssup_{\rho \in \mathcal{T}(\G)} \E\Bigl[\int_{\rho}^T \lambda_s ds \, | \, \mathcal F_{\rho}\Bigl] < + \infty
\]
Consequently, $\mathbb P(\tau\in [0,T])<1$, the support of $\tau$ strictly contains $[0,T]$.\\

\noindent \textit{Hypothesis B - bounded default}
\begin{equation} \label{H:bounded}
    \tag{HB}
    \esssup_{\rho \in \mathcal{T}(\G)} \E\Bigl[\int_{\rho}^t \lambda_s ds \, | \,  \mathcal G_{\rho}\Bigl] < + \infty \quad \forall t < T, \qquad \E\Bigl[ \Lambda_T \Bigl] = \infty.
\end{equation}
Hence, \[
    \esssup_{\rho \in \mathcal{T}(\G)} \E\Bigl[\int_{\rho}^t \lambda_s ds \, | \,  \mathcal F_{\rho}\Bigl] < + \infty \quad \forall t < T, \qquad \E\Bigl[ \Lambda_T \Bigl] = \infty.
\]
Consequently, $\mathbb P(\tau\in [0,T])=1$, the support of $\tau$ is included in $[0,T]$.

\subsection{Admissible strategy and contracts}
 Admissible strategies $\pi$ may be restricted to a closed subset of $\mathbb{R}^m$. We set the rigorous definition of an admissible strategy below, following \cite[Definition 1]{hu2005utility}.
 \begin{definition}[Admissible strategy with constraints]\label{admissible}
     Let $C$ be a closed and convex set in $\mathbb R^{1\times m}$. The set of admissible strategies, denoted by $\mathcal A$, consists of $m$-dimensional $\mathbb F-$predictable process $\pi$ such that $\mathbb E[\int_0^T|\pi_t|^2dt]<\infty$ and $\pi_t\in C$, $dt\otimes \mathbb P-$a.e.
 \end{definition}
 Note that due to the nature of the problem considered and by considering a compensation $\xi$, other integrability conditions are transferred to the admissibility of the contract below. 
Some examples of admissible strategy sets include:
\begin{itemize}
    \item $C = [0,1]^m$, that is $\pi$ is a proportion of the total wealth $X^\pi$ that is invested in the portfolio with no possibility to borrow or spend more than the actual value of $X^\pi$. It does not permit shorting stocks (i.e. selling stocks borrowed).
    \item $C =[-1,1]^m$, which does not permit leveraged positions.
    \item $C=[-M,M]^m$ assuming that the investor can spend or borrow as much money as needed limited to a fraction $M$ of the total wealth (possibly greater or less than 1 or $-1$).
\end{itemize}

For a symmetric positive definite matrix $Q$, we define the norm of a column vector $x$ by 
\[
    ||x||_{Q} = x^T Q x.
\]
This norm is equivalent to the euclidean norm in $\mathbb{R}^n$, and the constants of equivalence are the smallest and biggest eigenvalues of the matrix $Q$.
From this definition of $Q$-norm, we define the $Q$-distance of $x$ to the set $C$ as 
\[
    dist_Q(x, C) := \inf_{y \in C} \{ ||x - y||_Q \}.
\]

The contract $\xi$ proposed by the investor follows the idea of \cite{cvitanic2017moral}. We denote by $\eta>0$ a risk aversion parameter for the portfolio manager with CARA exponential utility function $U_A(x)=-e^{-\eta x},$ and assume that the manager may have to track an index and stay close to a benchmark strategy $\alpha$ over time given by a deterministic and bounded function of time with penalty coefficient $\varepsilon\geq 0$. A more rigorous definition is given in Section \ref{subseciton:delegatedmanagement} below. 

\begin{definition}[Admissible contract with contractible variables]\label{def:contract} We denote by $\Xi$ the set of admissible contracts $\xi$ composed of $\mathcal G_{\tau\wedge T}-$measurable random variable $\xi=Y_{T\wedge \tau}^{Y_0,Z,Z^X,U,\Gamma^X,\Gamma}$, controlled by $\mathbb G-$predictable real-valued processes $U,Z=(Z^i)_{1\leq i\leq m},Z^X,\Gamma^X,\Gamma=(\Gamma_i)_{1\leq i\leq m}$ such that $\mathcal{I}_m - \Gamma^X_t \sigma_t \sigma^T_t$ is a positive definite matrix\footnote{$\mathcal{I}_m$ denotes the identity matrix in dimension $m$.} and
    \begin{align*}
   Y_t^{Y_0,Z,Z^X,U,\Gamma^X,\Gamma}&= Y_0+ \int_{0}^{t}\sum_{i=1}^m \frac{Z^i_r}{S^i_r} dS^i_r + \int_{0 }^{t} Z_r^X dX_r + \int_{0}^{t} U_r dH_r \\
    &+ \frac{1}{2}\int_{0}^{t} (\Gamma_r^X + \eta (Z_r^X)^2 ) d\langle X, X\rangle_r + \int_{0}^{t} \sum_{i=1}^m \frac{\Gamma_r^i}{S^i_t} d\langle S^i, X\rangle_r\\
    &- \int_{0}^{t} F(r, Z_r, Z_r^X, \Gamma_r, \Gamma_r^X, U_r) dr,
\end{align*}
where \[ F(t, z, z_x, g, g_x, u)=\sup_{\nu\in C} f(t,z,z_x,g,g_x,u,\nu),\] with $f:[0,T]\times\Omega\times \mathbb R^m\times\mathbb R\times \mathbb R^m\times\mathbb R\times \mathbb R\times C$ by
\begin{align*} f(t,z,z_x,g,g_x,u,\nu)&= z b_t +z_x\nu \sigma_t \theta_t + \frac{1}{2} g_x\| \nu\sigma_t\|^2 - \frac{\varepsilon}{2}\|\nu - \alpha_t\|^2 +\sum_{i=1}^m g^i \nu^i \sigma_t^i(\sigma_t^i)^\top\\
&- \frac{\lambda_t}{\eta} (\exp{(-\eta u) - 1})- \frac{\eta}{2}||z \sigma_t||^2,\end{align*}
 so that
\begin{equation}\label{eq:generator}
    F(t, z_x, z, g_x, g, u) = \frac{1}{4}q_t^T Q_t^{-1} q_t + z b_t - \frac{\varepsilon}{2} ||\alpha_t||^2 - \frac{\lambda_t}{\eta} (\exp{(-\eta U_t) - 1}) - \frac{\eta}{2}||z \sigma_t||^2+ dist^2_Q(d_t, C)
\end{equation}
where 

 \[ Q_t = \frac{1}{2}(\mathcal{I} - g_x \sigma_t \sigma^T_t ),\quad  q_t = (\sigma_t \theta_t + \sigma_t \sigma^T_t g + \varepsilon\alpha_t - \eta z_x  \sigma_t \sigma^T_t z)\]
\[
    d_t = \frac{1}{2} Q_t^{-1} q_t, 
\]
and there exists $\eta'>\eta$ such that   \[\mathbb E\Big[\int_0^T (\|Z_s\|^2+\|Z^X_s\|^2+\|\Gamma_s\|+\|\Gamma^X_s\| +\|U_s\|^2\lambda_s) ds + \sup_{0\leq t\leq T} e^{\eta'|Y_t^{Y_0,Z,Z^X,U,\Gamma,\Gamma^X}|}\Big]<\infty.\]

The set of processes $Z,Z^X,U,\Gamma,\Gamma^X$ satisfying this integrability condition is denoted by $\mathcal U$ while their restriction to $[0,t]$ is denoted by $\mathcal U_t$ for any $t<T$.
   \end{definition}

    \paragraph{Economic interpretation.} 

    \begin{itemize}
        \item $Y_0$ is a fixed compensation determined by the reservation utility of the portfolio manager;
        \item the integrand process $Z^i$ represents a compensation with respect to the evolution of the $i$-th asset $S^i$, if it is observable by the investor and so contractible;
        \item $Z^X$ is a compensation term with respect to the portfolio dynamics, always observable by the investor;
        \item $U$ is a compensation with respect to the default risk of the market;
        \item $\Gamma^i$ is a compensation with respect to the covariation of $S^i$ and $X$ while the term $\Gamma^X+\gamma(Z^X)^2$ is a compensation driven by the risk aversion of the manager with respect to the quadratic variation of the portfolio;
        \item $F$ is the certainty equivalent utility gained by the portfolio manager when solving her optimization problem. The gain resulting from this optimization is transferred into the contract.
    \end{itemize}

       \begin{remark}
A rigorous justification of this set of contracts is given in Appendix \ref{app:2bsdecontract}. This appendix provides a rigorous framework for defining the probability space and the weak formulation of the portfolio manager's problem together with its solution by using second-order BSDEs and the recent results in \cite{gennaro20252bsde}. 
        \end{remark}
    \begin{remark}
We can refine the set of contracts depending on the information available to the investor (see \cite{cvitanic2017moral}).
    \begin{itemize}
         \item We denote by $\Xi^\circ\subset \Xi$ the set of random variables $\xi=Y_T^{Y_0,Z^X,U,\Gamma^X}$ with $Z^i=\Gamma^i=0$
corresponding to the case where $S$ is not observable by the investor and so not contractible.
\item The set of linear contracts $\Xi^l\subset \Xi$ defined by 
\begin{align*}
   \xi &=Y_0+ \int_{0}^{\T} \sum_{i=1}^m \frac{Z^i_s}{S^i_s} dS^i_s + p(X_{\T} - X_0) + \int_{0}^{\T} U_s dH_s + \frac{1}{2}\int_{0}^{\T} (\Gamma_s^X + \eta p^2 ) d\langle X, X\rangle_t\\
  & - \int_{0}^{\T} F(s, Z_s, p, \Gamma_s, \Gamma_s^X,U_s) ds+ \int_{0}^{\T} \sum_{i=1}^m \frac{\Gamma_s^i}{S^i_s} d\langle S^i, X\rangle_t 
\end{align*}
with contractible $S$. The idea around this contract is that in practice, most of fund managers ask as a compensation to their clients a fixed percentage of the terminal wealth, forcing $Z_s^X = p$ in $\Xi$, where $p$ is the fixed percentage the agent receives.
\end{itemize}

In all these cases, $Z=(Z^i)_{1\leq i\leq m},Z^X,U,\Gamma,\Gamma^X$ are predictable processes such that $\xi\in \mathcal C$ and all the stochastic integrals are martingales.  
    \end{remark}

\begin{remark}
    Note that the set of contracts $\Xi$ is stated without loss of generality as explained in \cite{cvitanic2018dynamic,chiusolo2024new} as soon as we consider general contracts $\xi$ with exponential moment of any order, requiring to refine the definition of $\mathcal U$ in order to apply the existence results in \cite{kharroubi2013mean,jeanblanc2015utility}. 
\end{remark}

\subsection{Delegated portfolio management and bi-level stochastic programming}\label{subseciton:delegatedmanagement}
We are assuming that the portfolio manager is receiving the contract $\xi$ and optimally chooses a strategy $\pi$ in order to remain close to a benchmark strategy $\alpha$ \textit{\`a la} Almgren-Chriss, see \cite{almgren2001optimal} so that the objective of the manager is to solve for any contract $\xi\in \Xi$ fixed
\begin{equation}\label{pbAgent}  \tag{A} V^A_0(x,\xi)=\sup_{\pi\in \mathcal A} J^A(\pi;x,\xi),\end{equation}
with
\[J^A(\pi;x,\xi):=\mathbb E[ U^A(\xi-\frac{\varepsilon}2 \int_{0}^{\T} \|\pi_s - \alpha_s\|^2 ds)],\]
and where $\varepsilon\geq 0$ represents the weight on the penalty. Note that when $\varepsilon=0$ we do not track any index, which is consistent with the cryptocurrency motivation. In contrast, taking $\varepsilon>0$ introduces index-tracking power (as in an ETF) and aligns with the pension-plan motivation. In our model, the investor chooses a full delegation of the portfolio management to the manager and thus lets the manager choose the optimal strategy to optimize the terminal value of the portfolio under default. In this \textit{second-best} case, the contracting problem is reduced to solving a bi-level optimization under constraint as follows when $S$ is contractible

\begin{equation}    \label{eq:principalPb}
    \tag{P}
     V_0^P(x) := \sup_{(\xi,\hat\pi) \in \mathcal C\times \mathcal{A}} \mathbb{E}[U^P(X_{T\wedge \tau}-\xi)]
\end{equation}
subject to
\begin{itemize}
    \item (R):\; $V_0^A(x,\xi)\geq R_0$
    \item (IC):\; $V_0^A(x,\xi) = J^A(\hat\pi;x,\xi)$.
\end{itemize}

We will refer to the first problem \eqref{pbAgent} as the Problem of the Agent, and the bi-level program (P) as the Problem of the Principal. 
Furthermore, we want to emphasize the fact that $\tau$ does not need to be a stopping time in the natural filtration given by the assets' dynamics. Despite this, most results regarding default times are based on stopping-time theory, so we want to work with a filtration such that $\tau$ can be a stopping time.

\begin{remark}[First-best solution]

Alternatively, one could assume that the investor controls both the allocation strategy and the compensation proposed to the portfolio manager. In this \textit{first-best} situation, the problem of the principal (investor) is reduced to an optimization of its projected P\&L over both $\xi$ and $\pi$ subject only to the reservation constraint. That is, solving:

\begin{equation*}
    V^{FB}_0 = \sup_{(\xi,\pi)} \mathbb{E} \left[ U^P(X^\pi_{T\wedge \tau} - \xi) \right],
\end{equation*}

subject to 

\begin{equation*}
    J^A(\pi; x, \xi) \geq R.
\end{equation*}

Note that $V^{FB}_0 \geq V^P_0$. We refer to Appendix \ref{FB-appendix} for the details and the solution of this first-best benchmark.
\end{remark}

\begin{remark}[Extension to Merton's problem with consumption]
A further extension of this problem considers an investor who chooses to consume a portion of the portfolio continuously over time. Here, the consumption strategy $c_t$ given by an $\mathbb F-$adapted random process. The specific mathematical derivations for this case are provided in Appendix \ref{consumption extension}.
\end{remark}

\section{Optimal contract and investment strategy}\label{chap3}

\subsection{Optimal investment with random horizon: solving the agent problem}
A common approach in the continuous stochastic optimization literature is based on the solution of Backward Stochastic Differential Equations by solving a martingale optimality principle. We refer the reader to \cite{hu2005utility} for a detailed explanation of the method in a continuous setting and to \cite{morlais2009utility,kharroubi2013mean,jeanblanc2015utility} for extensions to discontinuous processes or default times. The general idea is to generate a family of super-martingales $(R^\pi)$ indexed by the control variable, in our case the investment strategy $\pi$, with terminal condition the objective function of the agent $J^A$ at time $T\wedge \tau$. If we are able to find a specific control $\hat\pi$ such that $R^{\hat\pi}$ is a martingale, this control is optimal for \eqref{pbAgent}, and the optimal value is given by the process corresponding to the optimal control at time $0$. 

\begin{lemma}[Martingale Optimality Principle] \label{theorem:MOP}
    Let $\xi\in \Xi$ and let $(R^{\pi}_{t\wedge \tau})_{t \in [0, T]}$ be a family of stochastic processes indexed by the strategy $\pi \in \mathcal A$ such that
    \begin{enumerate}
        \item[(i)] $R^{\pi}_{T\wedge \tau} = U^A(\xi- \frac{\varepsilon}2\int_{0}^{\T} \|\pi_s - \alpha_s\|^2 ds)$,
        \item[(ii)] $R^{\pi}$ is a $\G$-supermartingale and $R^{\pi}_0 $ is constant for all $\pi \in \mathcal A$,
        \item[(iii)] there exists $\hat\pi\in \mathcal A$ such that $R^{\hat\pi}$ is a $\G$-martingale.
    \end{enumerate}
    Then, $\hat\pi$ is a solution to the maximization problem \eqref{pbAgent}.
\end{lemma}
\begin{proof}
    Take $\pi \in \mathcal A$. Then, we have
    \begin{align*}
       J^A(\pi;x,\xi) = \E[R^{\pi}_{\T}] \overset{(i)}\leq \E[R^{\pi}_0]\overset{(ii)}  =  \E[R^{\pi^*}_0] \overset{(iii)} =  \E[R^{\pi^*}_{\T}] = J^A(\hat\pi;x,\xi).
    \end{align*}\end{proof}
Let $\xi\in\Xi$ be fixed. 
Independently of the boundedness of the default time, that is either under Hypothesis A or Hypothesis B, we define
\[ R^{\pi}_{t\wedge \tau} := U^A(Y^\pi_{t\wedge\tau} - \int_0^{t\wedge\tau} \|\pi_s - \alpha_s\|^2 ds),\]
   where $Y^\pi$ is defined by
\begin{align*} \label{eq:BSDE}
    Y^\pi_{t \wedge \tau} &=Y_0+ \int_{0}^{t \wedge \tau} \sum_{i=1}^m \frac{Z^i_s}{S^i_s} dS^i_s + \int_0^{t \wedge \tau} Z_s^X dX_s + \int_{0}^{t \wedge \tau} U_s dH_s \\
    &+ \frac{1}{2}\int_0^{t \wedge \tau} (\Gamma_s^X + \eta (Z_s^X)^2 ) d\langle X, X\rangle_s - \int_0^{t \wedge \tau} F(s, Z_s, Z_s^X, \Gamma_s, \Gamma_s^X,U_s) ds\\
    &+\int_0^{t \wedge \tau} \sum_{i=1}^m \frac{\Gamma^i_s}{S^i_s} d\langle S^i, X\rangle_s
\end{align*}

\begin{theorem}\label{thm:value}
    Assume that \ref{densityhyp} and either Hypothesis A or Hypothesis B are satisfied. For any $\xi\in \Xi$, the optimal strategy that solves \eqref{pbAgent} is 
  \begin{equation}\label{eq:optStrat}
        \hat\pi_t =\pi^*(Z_t, Z_t^X, \Gamma_t, \Gamma^X_t),\text{ with }\pi^*(z,z_x,g,g_x):= proj(e_t, C) \end{equation}

    and the optimal value is given by 
        $V^A_0(x,\xi) = -e^{-\eta Y_0}$, where
    \[
        e_t = \frac{1}{2}q_t Q_t^{-1},\; q_t = (\sigma_t \theta_t + \sigma_t \sigma^T_t \Gamma_t + \varepsilon\alpha_t - \eta Z_t^X  \sigma_t \sigma^T_t Z^T_t)
    \]

    and
    \[
        Q_t = \frac{1}{2}(\mathcal{I} - \Gamma^X_t \sigma_t \sigma^T_t ).
    \]
\end{theorem}

\begin{proof}

    The proof follows the same idea as \cite{hu2005utility} for the continuous case and \cite{jeanblanc2015utility, morlais2009utility} for the discontinuous case extended to the multi-dimensional case and the contract $\xi$ fixed by the principal. Note that $Z = (Z^i)_{i=1}^m$ and $\Gamma = (\Gamma^i)_{i=1}^m$ are already $\mathbb{R}^m$ stochastic row vectors and that we are going to denote by $C$ the set of admissible trading strategies $\pi = (\pi_t)_{t \in [0, T]}$. The proof will be based on Ito's formula with Poisson jumps (as the intensity of our jump is the same as that of a simple Poisson process with varying intensity $\lambda$) and we refer to \cite{privault2022introduction,ikeda2014stochastic}, for more details on stochastic calculus with semi-martingales; and the notion of \emph{Doléans-Dade exponential process} (DDE). Given a semimartingale $P$, we denote its DDE by 
        \[
            \mathcal{E}(P)_t = \exp\{ P_t - \frac{1}{2} \langle P^c \rangle_t\} \prod_{0 \leq s \leq t}\{(1 + \Delta_s P) \exp( -\Delta_s P)\},
        \]
        where $P^c$ is the continuous component of the path of $P$ while $\Delta_s P:=P_s-P_{s-}$.
        Note that the DDE of $P$ is a solution of the following SDE $dY_t = Y_{t-} dX_t$
        and that, provided $P$ is a $BMO$-martingale, such that its jump is greater than $-1+\delta$, with $\delta > 0$, the resulting DDE is a martingale as well.\\

To find the optimal solution, we are applying the so-called martingale optimality principle. We define the following family of stochastic processes indexed by the strategy $\pi$
    \begin{equation*}
        R^{\pi}_{t} = U^A(Y_t - \frac{\varepsilon}2\int_0^t ||\pi_s - \alpha_s||^2 ds).
    \end{equation*}
    We set $B^\pi_t := Y_t - \frac{\varepsilon}2\int_0^t ||\pi_s - \alpha_s||^2 ds$, so that $R^\pi_t = -\exp{(-\eta B^\pi_t)}$. Note that the Ito's decomposition of $B^\pi$ is given by
    \begin{align*}
        dB^\pi_t 
        &= \sum_{i=1}^{m}\frac{Z_t}{S^i_t} dS^i_t + Z_t^X dX_t + U_t dH_t + \frac{1}{2} (\Gamma_t^X + \eta (Z_t^X)^2 ) d\langle X, X\rangle_t\\
        &- (F(t, Z_t, Z_t^X, \Gamma_t, \Gamma_t^X,U_t) + \frac{\varepsilon}2||\pi_t - \alpha_t||^2) dt + \sum_{i=1}^{m}\frac{\Gamma^i_t}{S^i_t} d\langle S^i, X\rangle_t,\; t \in [0, \T].
    \end{align*}

Hence,
    \begin{align*}
        dB^\pi_t&=J^\pi_t dt+ (Z_t \sigma_t + Z^X_t \pi_t \sigma_t)dW_t+ U_t dH_t,
    \end{align*}
where 
\[J^\pi_t= Z_t b_t + Z^X_t \pi_t \sigma_t \theta_t + \frac{1}{2}(\Gamma^X_t + \eta(Z^X_t)^2)|| \pi_t \sigma_t||^2 - F(t, Z_t, Z_t^X, \Gamma_t, \Gamma_t^X,U_t)  + \pi_t \sigma_t \sigma^T_t \Gamma_t - \frac{\varepsilon}{2}|| \pi_t - \alpha_t||^2).  \]
We thus deduce that
    \begin{align*}
        dR^\pi_t&= -\eta R^\pi_t  J^\pi_t dt+ \frac{\eta^2}{2} R^\pi_t (Z_t \sigma_t + Z^X_t \pi \sigma_t)dW_t+ R^\pi_{t_{-}} (\exp{(-\eta U_t) - 1}) dM_t
    \end{align*}
    where we used the fact that $dH_t = dM_t + \lambda_t dt$. Note that $R_0 = -\exp{(-\eta Y_0)} < 0$ and $J^\pi\leq 0$ and is equal to zero when $\pi$ is chosen by maximizing
    \[
        -\pi_t Q_t \pi_t^T + \pi_t q^T + d_t
    \]
    with
    \[
        Q_t = \frac{1}{2}(\mathcal{I} - \Gamma^X_t \sigma_t \sigma^T_t ),
    \]
    \[
        q_t = (\sigma_t \theta_t + \sigma_t \sigma^T_t \Gamma_t + \varepsilon \alpha_t - \eta Z_t^X  \sigma_t \sigma^T_t Z^T_t),
    \]
    and
    \[
        d_t = Z_t b_t - F - \frac{\eta}{2} || Z_t \sigma_t ||^2 - \frac{\varepsilon}{2} ||\alpha_t||^2 - \frac{\lambda_t}{\eta} (\exp{(-\eta U_t) - 1}).
    \]

  Therefore,
    \[
        \pi^*(Z_t, Z_t^X, \Gamma_t, \Gamma^X_t) = proj(e_t, C)
    \]
    so that $\pi^*$ is optimal and is unique because our set $C$ is closed and convex.
    \end{proof}

\begin{remark}
    Note that if $m=d=1$ and $C=\mathbb R$ and $\sigma$ and $b$ are constant we obtain 
    \begin{align*}
        F(t, z, z_x, g, g_x, u) 
        &= \frac{1}{2} \Bigl(\frac{z_x\sigma \theta - \varepsilon \alpha_s -\eta z z_x \sigma^2 + \sigma^2 g}{\sqrt{1 - g_x \sigma^2}}\Bigl)^2 + zb \\
        &- \frac{\varepsilon}{2}\alpha_t^2 - \frac{\lambda_t}{\eta}((e^{-\eta u} - 1)) - \frac{\eta}{2} z^2 \sigma^2.
    \end{align*}
\end{remark}
\begin{remark}
As a consequence of the constraint (IC), we are implicitly looking for contracts such that  $\hat\pi\in \mathcal U$. This requires to consider $(Z,Z^X,\Gamma,\Gamma^X,U)$ such that $\hat\pi=\pi^*(Z, Z^X, \Gamma, \Gamma^X)\in \mathcal U$, otherwise (IC) is not satisfied and there are no optimal contracts. 
\end{remark}

\begin{remark}[Extension to CRRA utilities]
We have shown the optimal strategy for the CARA utility. We extend the analysis to power CRRA utilities in Appendix \ref{AppendixPower}. While for general utilities, we might not be able to find an explicit formula for the driver F and the optimal strategy $\pi$, it is possible to find them at least numerically. 
\end{remark}

\subsection{The optimal contract and verification results}\label{chap4}
In the previous chapter, we found the optimal strategy $\pi^*$ for the agent under the best optimal contract with contractible $S$. Given that strategy, we now want to solve the problem for the principal of \eqref{eq:principalPb}, i.e., we want to find the compensation $\xi$ to maximize the principal utility:
\begin{equation}\label{principalpb}
   V_0= \sup_{\xi\in \Xi \text{ s.t } V_0^A(x) \geq R}\mathbb{E}[U^P(X_{\T}^{\pi^*} - \xi)]
\end{equation}

where we used $X_{\T}^{\pi^*}$ to stress that now the object of our maximization problem is dependent on the agent’s best strategy $\pi^*$. Before diving into the intuition behind how to solve this problem, it is important to show how this becomes a control problem so that an HJB equation can be derived. To do so, we first want to work with $\xi \text{ s.t } V_0^A(x) \geq R$. 
We recall that $V_0^A(x) = -\exp( - \eta( Y_0))$, so that
\[ V_0^A(x) \geq R \Longleftrightarrow  Y_0 \geq - \frac{\log(-R)}{\eta} := \widehat{Y}_0.
\]

Moreover,
\begin{align}\label{xi:decomp}
   \nonumber \xi= Y^\pi_{\T} &=Y_0+ \int_{0}^{\T} \sum_{i=1}^m \frac{Z^i_t}{S^i_t} dS^i_t + \int_0^{\T} Z_t^X dX_t + \int_{0}^{t \wedge \tau} U_s dH_s \\
     \nonumber &+ \frac{1}{2}\int_0^{\T} (\Gamma_t^X + \eta (Z_t^X)^2 ) d\langle X, X\rangle_t - \int_0^{\T} F(t, Z_t, Z_t^X, \Gamma_t, \Gamma_t^X,U_t) dt\\
      &+\int_0^{\T} \sum_{i=1}^m \frac{\Gamma^i_s}{S^i_t} d\langle S^i, X\rangle_t 
\end{align}
Therefore, we can rewrite \eqref{principalpb} as
  \begin{align*} 
    V_0=\sup_{Y_0\geq \hat R_0} \sup_{(Z, Z^X, \Gamma, \Gamma^X, U)\in \mathcal U} \mathbb{E}[U^P(X_{\T}^{\pi^*} -Y^{\pi^*}_{\T})].
\end{align*}
Since $U^P$ is nondecreasing in $Y_0$, we deduce that 

  \begin{align} \label{eq:preHJB}
    V_0= \sup_{(Z, Z^X, \Gamma, \Gamma^X, U)\in \mathcal U} \mathbb{E}[U^P(X_{\T}^{\pi^*} -\hat Y^{\pi^*}_{\T})],
\end{align}
where 
$\hat Y^{\pi^*}$ is given by \eqref{xi:decomp} with $Y_0=\hat R_0:=-\log(-R)/\eta.$

\begin{remark}[Risk neutral investor]
    Assume that $U^P(x)=x$, then we note that 
    \[
    V_0= \hat V_0 - \hat R_0,
 \]
 where \eqref{eq:preHJB} becomes

     \begin{align*}\hat V_0:=\sup_{(Z, Z^X, \Gamma, \Gamma^X, U)\in \mathcal U} \mathbb{E}[U^P(X_{\T}^{\pi^*} -\hat Y^{\pi^*}_{\T})],
\end{align*}
with 
\[
\begin{cases}
     &d\hat Y^{\pi^*}_t= \sum_{i=1}^m \frac{Z^i_t}{S^i_t} dS^i_t +  Z_t^X dX^{\pi^*}_t +  U_t dH_t+ \frac{1}{2} (\Gamma_t^X + \eta (Z_t^X)^2 ) d\langle X^{\pi^*}, X^{\pi^*}\rangle_t\\
     &\qquad - F(t, Z_t, Z_t^X, \Gamma_t, \Gamma_t^X,U_t) dt+\sum_{i=1}^m \frac{\Gamma^i_t}{S^i_t} dS^i_t d\langle S^i, X^{\pi^*}\rangle_t\\
     &\hat Y_0^{\pi^*}=0.
     \end{cases}\]
\end{remark}
In order to derive the HJB equation with control process $(Z, Z^X, \Gamma, \Gamma^X, U)$, we will use the results from \cite{blanchet2008optimal} assuming that $\tau$ has a probability density function $f$ and a cumulative distribution function $F$, whether the support of $\tau$ is bounded in $[0,T]$ or unbounded in the next two subsections. To derive the HJB equation from this problem, we set an additional assumption enforcing Markovian properties of the intensity $\lambda$, drift and volatility processes.\\

\noindent \textbf{Assumption (M)}. We assume that there exist two $\F-$progressive measurable functions $b,\sigma,\lambda$ such that $b_t=b(t,S_t),\sigma_t=\sigma(t,S_t)$ and $\lambda_t=(t,S_t,X_t)$.\\

Note that if $\tau$ is independent of $\mathbb F$, as we will consider for numerical simulations, the density of $\tau$, and hence $\lambda$, becomes a deterministic function for computations.\\

We define the following system of coupled SDEs with jumps with solution $(X^{\pi^*},Y^{\pi^*})$ controlled by $(Z,Z^X,\Gamma,\Gamma^X,U)$

\[
(SDE)\begin{cases}
     &dX^{\pi^*}_t =  \pi_t^*\sigma(t,X_t) dW_s + \pi_t^* b(t,S_t) dt\\
     &d\hat Y^{\pi^*}_t= \sum_{i=1}^m \frac{Z^i_t}{S^i_t} dS^i_t +  Z_t^X dX^{\pi^*}_t +  U_t dH_t+ \frac{1}{2} (\Gamma_t^X + \eta (Z_t^X)^2 ) d\langle X^{\pi^*}, X^{\pi^*}\rangle_t\\
     &\qquad - F(t, Z_t, Z_t^X, \Gamma_t, \Gamma_t^X,U_t) dt+\sum_{i=1}^m \frac{\Gamma^i_t}{S^i_t} dS^i_t d\langle S^i, X^{\pi^*}\rangle_t\\
     &X^{\pi^*}_0=x,\\
     &Y^{\pi^*}_0=\hat R_0. 
\end{cases}
\]

\subsubsection{Bounded default time}
Under Hypothesis B, the support of $\tau$ is included in $[0,T]$. Recalling the results in \cite{blanchet2008optimal}, we deduce that

\begin{equation*}
V_0= \sup_{(Z, Z^X, \Gamma, \Gamma^X, U)\in \mathcal U}\;  \E\bigl[ \int_{0}^T U^P(X_t^{\pi^*} -\hat{Y}^{\pi^*}_t) f(t)dt \bigl],
\end{equation*}
where $f$ is the density of $\tau$. We introduce the following Hamilton-Jacobi-Bellman integro-partial differential equation.

\[
\textbf{(bHJB)}\quad\begin{cases}
  &  \partial_t \phi(t,s,x,y) + U^P(x - y) f(t) + \sup_{(z^x, z, g^x, g, u)}\bigl \{   \mathcal H^\phi(t,s, x, y,  z^x, z, g^x, g, u) \bigl \}=0,\; t<T\\
&\phi(T,s,x,y)=0,\; (s,x,y)\in \mathbb R^m\times \mathbb R\times \mathbb R.
\end{cases}
\]
 where $\mathcal H^\phi$ is a differential operator given by
\begin{align*}
&\mathcal H^\phi(t,s,x,y,z^x,z,g^x,g,u):=\\
&\sum_{i=1}^m\phi_{s^i} s^ib^i(t,s)  + \phi_x \pi^*(z,z^x,g,g^x) \sigma(t,s) \theta(t,s) \\
&+ \frac12 Tr(\Sigma(t,s)\Sigma(t,s)^\top  D^2\phi)+ \frac{1}{2}\phi_{xx}||\pi^*(z,z^x,g,g^x)  \sigma(t,s)||^2\\
& + \phi_y \Big[z b(t,s) + z^x \pi^*(z,z^x,g,g^x) \sigma(t,s) \theta(t,s)
+ \frac{1}{2}||\sigma(t,s) \pi^*(z,z^x,g,g^x)||^2(g^x + \eta |z^x|^2) \\
&-  F(t, z, z^x, g, g^x,u) + \pi^*(z,z^x,g,g^x) \sigma(t,s) \sigma(t,s)^T g\Big]
   \\
&+ \frac{1}{2}\phi_{yy} (||z_t \sigma(t,s)||^2 + z^x||\pi^*(z,z^x,g,g^x) \sigma(t,s)||^2) \\
&+ \phi_{xy} (z^x ||\pi^*(z,z^x,g,g^x) \sigma(t,s)||^2  + \pi^*(z,z^x,g,g^x) \sigma(t,s) (z \sigma)^T)\\
&+ \sum_{i=1}^m (\pi^*(z,z^x,g,g^x)  \sigma(t,s) \sigma^i(t,s)^T s^i)(\phi_{xs^i} + \phi_{ys^i}z^x) + \lambda_t(\phi(t, s,x, y+u) - \phi(t,s, x, y))\\
&+\sum_{i=1}^m \phi_{ys^i} z^i \sigma(t,s) \sigma^i(t,s)^T s^i,
\end{align*}
with the notation $\Sigma^i(t,s)=s^i \sigma^i(t,s)$ is the $i$-th row of the matrix $\Sigma$.
We define 
\[\hat Z_t:=\hat z(t,S_t,X^{\pi^*}_t,\hat Y^{\pi^*}_t),\; \hat Z_t^X=\hat z_x(t,S_t,X^{\pi^*}_t,\hat Y^{\pi^*}_t),\; \hat U_t=\hat u(t,S_t,X^{\pi^*}_t, \hat Y^{\pi^*}_t),\]
\[\hat \Gamma_t=\hat g(t,S_t,X^{\pi^*}_t,\; \hat Y^{\pi^*}_t), \hat \Gamma_t^X=\hat g_x(t,S_t,X^{\pi^*}_t, \hat Y^{\pi^*}_t),\]
where $\hat z,\hat z_x,\hat g,\hat g_x,\hat u$ optimize $\mathcal H^\phi$.
\begin{theorem}[Verification Theorem - bounded case]\label{thm:verif}
Assume that there exists a function $\phi$ twice continuously differentiable in space and differentiable in time, such that $\phi(t,s,x,y)$ solves \textbf{(bHJB)}. Furthermore, assume that $\phi$ has a quadratic growth in $y$ and polynomial growth in $s,x$ such that 
\[|\phi(t,s,x,y)|\leq \kappa(1+|x|^p+\|s\|^p+|y|^2),\; p>1,\; \kappa>0.\]

Then, for each $t \in [0,T]$, the strategy $(\hat{Z}^X, \hat{Z}, \hat{\Gamma}^X, \hat{\Gamma}, \hat{U})$ is an optimal strategy for the control problem 

\begin{align*}
    \phi(0,S_0,x,0)=V_0 =  \sup_{(Z^X, Z, \Gamma^X, \Gamma, U)} \mathbb{E}\left[U^P(X_{T \wedge \tau}^{\pi^*} - \hat Y^{\pi^*}_{T \wedge \tau})\right].
\end{align*}

The optimal contract is given by 
\begin{align*}
    \xi^\star 
    &= Y_0+ \int_{0}^{\T} \sum_{i=1}^m \frac{\hat{Z}^i_t}{S^i_t} dS^i_t + \int_0^{\T} \hat{Z}_t^X dX_t + \int_{0}^{\T} \hat{U}_s dH_s \\
    &+ \frac{1}{2}\int_0^{\T} (\hat{\Gamma}_t^X + \eta (\hat{Z}_t^X)^2 ) d\langle X, X\rangle_t - \int_0^{\T} F(t, \hat{Z}_t, \hat{Z}_t^X, \hat{\Gamma}_t, \hat{\Gamma}_t^X,U_t) dt\\
    &+\int_0^{\T} \sum_{i=1}^m \frac{\hat{\Gamma}^i_s}{S^i_t} d\langle S^i, X\rangle_t 
\end{align*}
\end{theorem}

\begin{remark}
  If we assume that $\sigma(t,s)=\sigma\in \mathbb R^{m\times d}$ and $b(t,s)=b\in\mathbb R^m$, we note that the solution $v$ to the HJB equations \textbf{(bHJB)} does not depend on $s$. It can thus be rewritten as \[
\textbf{(bHJB)}\quad\begin{cases}
  \partial_t v(t,x,y) + U^P(x - y) f(t) + \sup_{(z^x, z, g^x, g, u)}\bigl \{   \mathcal H^v(t, x, y, \nabla v, \Delta v, z^x, z, g^x, g, u) \bigl \}=0,\; t<T\\
v(T,x,y)=0,\; (x,y)\in\mathbb R\times \mathbb R,
\end{cases}
\]
where

\begin{align*}
&\mathcal H^\phi(t,x,y,z^x,z,g^x,g,u):= \phi_x \pi^*(z,z^x,g,g^x)\sigma \theta + \frac{1}{2}\phi_{xx}||\pi^*(z,z^x,g,g^x) \sigma||^2\\
& + \phi_y \Big[z^T b + z^x \pi^*(z,z^x,g,g^x) \sigma \theta
+ \frac{1}{2}||\pi^*(z,z^x,g,g^x)\sigma||^2(g^x + \eta |z^x|^2) \\
&\qquad-  F(t, z, z^x, g, g^x,u) + \pi^*(z,z^x,g,g^x) \sigma \sigma^T g\Big]
   \\
&+ \frac{1}{2}\phi_{yy} (||z^T\sigma||^2 + z^x|| \pi^*(z,z^x,g,g^x)\sigma||^2) \\
&+ \phi_{xy} (z^x ||\pi^*(z,z^x,g,g^x) \sigma||^2  + \pi^*(z,z^x,g,g^x)\sigma\sigma^T  z)\\
&+\lambda_t(\phi(t, x, y+u) - \phi(t, x, y)).
\end{align*}
\end{remark}

\begin{proof}
The proof of the theorem relies on a localization procedure. We first assume that there exists a solution $\phi$ to \textbf{(bHJB)} satisfying \[|\phi(t,s,x,y)|\leq \kappa(1+|x|^p+\|s\|^p+|y|^2),\; p>1,\; \kappa>0.\]  Let $(Z,Z^X,\Gamma,\Gamma^X,U)\in \mathcal U$ We introduce the following stopping time
\[
    \tau_n := \inf\{ t > 0, (S_t, X^{\pi^*}_t, \hat Y^{\pi^*}_t, Z^X,\lambda_t) \in \mathcal{B}_n(0)\} \wedge T
\]
where $\mathcal{B}_n(0)$ is the ball centered at the origin with radius $n$ in $\mathbb{R}^{m+4}$.
Applying Ito's formula, we get
\begin{align*}
    \phi(\tau_n, S_{\tau_n}, X^{\pi^*}_{\tau_n}, \hat Y^{\pi^*}_{\tau_n}) 
    &= \phi(0, S_{0}, x, 0) + \int_{0}^{\tau_n} \mathcal H^\phi(r,S_r,X^{\pi^*}_r,\hat Y^{\pi^*}_r,Z^X_r,Z_r,\Gamma^X_r,\Gamma_r,U_r) dr \\
    &+\int_{0}^{\tau_n} \sum_{i=1}^m \partial_{s^i} \phi(r,S_r,X^{\pi^*}_r, \hat Y^{\pi^*}_r) S_r^i \sigma^i_r dW_r\\
    &+\int_{0}^{\tau_n} \partial_{x} \phi(r,S_r,X^{\pi^*}_r, \hat Y^{\pi^*}_r) \pi_r^* \sigma_r dW_r\\
    &+\int_{0}^{\tau_n} \partial_{y} \phi(r,S_r,X^{\pi^*}_r, \hat Y^{\pi^*}_r) (\sum_{i=1}^m Z^i_r \sigma^i_r +Z_r^X\pi_r^* \sigma_r )dW_r\\
    &+\int_{0}^{\tau_n} \lambda_r(\phi(r,S_r,X^{\pi^*}_r, \hat Y^{\pi^*}_r + U_r) - \phi(r,S_r,X^{\pi^*}_r,\hat Y^{\pi^*}_r)) dM_r.
\end{align*}
By the localization procedure and considering the expectations, we get
\begin{align*}
    &\mathbb{E}[\phi(\tau_n, S_{\tau_n}, X^{\pi^*}_{\tau_n}, \hat Y^{\pi^*}_{\tau_n})]\\
    &= \phi(0, S_{0}, x, 0)- \mathbb{E}[\int_{0}^{\tau_n}U^P(X^{\pi^*}_r - \hat Y^{\pi^*}_r)f(r) dr]\\
    &+\mathbb{E}[\int_{0}^{\tau_n} \mathcal (\mathcal H^\phi(r,S_r,X^{\pi^*}_r,\hat Y^{\pi^*}_r,Z^X_r,Z_r,\Gamma^X_r,\Gamma_r,U_r) + \partial_t \phi(r,S_r,X^{\pi^*}_r,\hat Y^{\pi^*}_r) + U^P(X^{\pi^*}_r -  \hat Y^{\pi^*}_r)f(r) )dr].
\end{align*}
Since $\phi$ satisfies \textbf{(bHJB)}, 
\begin{align*}
    \mathbb{E}[\phi(\tau_n, S_{\tau_n}, X^{\pi^*}_{\tau_n}, \hat Y^{\pi^*}_{\tau_n})+\int_{0}^{\tau_n}U^P(X^{\pi^*}_r -  \hat Y^{\pi^*}_r)f(r) dr] 
    &\leq \phi(0, S_{0}, x, 0) \quad \forall (Z^X,Z,\Gamma^X,\Gamma,U)\in \mathcal U
\end{align*}
where the equality holds for $(\hat{Z}^X,\hat{Z},\hat{\Gamma}^x,\hat{\Gamma},\hat{U})$.Recall that $\phi$ is bounded by a polynomial function in the variables $(s, x, y)$. We note that \[
    |\phi(\tau_n, S_{\tau_n}, X^{\pi^*}_{\tau_n}, \hat Y^{\pi^*}_{\tau_n})| \leq C(1 + |S_{\tau_n}|^p + |X^{\pi^*}_{\tau_n}|^p + | \hat Y^{\pi^*}_{\tau_n}|^2)
    \leq C(1 + \sup_{t} \{|S_t|^p + |X^{\pi^*}_{t}|^p + | \hat Y^{\pi^*}_{t}|^2 \}).
\]By applying the dominated convergence theorem, we deduce that
\[
   \mathbb{E}[\phi(T, S_{T}, X^{\pi^*}_{T}, \hat Y^{\pi^*}_{T})+\int_{0}^{T}U^P(X^{\pi^*}_r -  \hat Y^{\pi^*}_r)f(r) dr]  \leq  \phi(0, S_0, x,0),
\]
with equality when $ (Z^X,Z,\Gamma^X,\Gamma,U)$ = $(\hat{Z}^X,\hat{Z},\hat{\Gamma}^X,\hat{\Gamma},\hat{U})$.

\end{proof}

\paragraph{Viscosity solution and dynamic programming.} Assuming the existence of a $\mathcal C^{1,2,2,2}$ solution to \textbf{(bHJB)} can be relaxed in Theorem \eqref{thm:verif} by showing that the value function of the problem is a solution in the sense of viscosity of the PDE \textbf{(bHJB)}. This relaxation of regularity for solution to PDEs has been developed by Lions and Crandall in \cite{crandall1983viscosity,crandall1992user}. We refer to  \cite{fleming2006controlled,touzi2012optimal} for more details about stochastic control with viscosity solutions. We start by recalling the dynamic programming principle and we define the continuation value objective of the principal for any control $\nu:=(Z,Z^X,\Gamma,\Gamma^X,U)$, by

\begin{equation*}
  v(t,s,x,y,\nu)=  \E\bigl[ \int_{t}^T U^P(X^{t,x,\pi^*}_r - Y^{t,y,\pi^*}_r) f(r)dr \bigl],
\end{equation*}
so that $V(t,s,x,y)=\sup_\nu v(t,s,x,y,\nu),$
where $X^{t,x,\pi^*}$ and $Y^{t,y,\pi^*}$ denote the flow processes starting at time $t$ with respective initial values $x$ and $y$ in (SDE). The dynamic programming principle states that
for any stopping time $\theta\in [t,T]$ where $t<T$ we have
\begin{equation}\label{DPP}
  V(t,s,x,y)= \sup_{(Z,Z^X,\Gamma,\Gamma^X,U)\in \mathcal U_{[t,\theta]}} \E\bigl[ V(\theta, S_\theta^{t,s},X_\theta^{t,x,\pi^*},Y_\theta^{t,y}) + \int_{t}^\theta U^P(X^{t,x,\pi^*}_r - {Y}^{t,y}_r) f(r)dr \bigl].
\end{equation}

The notion of weak solutions of \textbf{(bHJB)} results from this dynamic programming principle and is defined as follows.

\begin{definition}[Viscosity solution]\label{viscosity:def} We say that $v$ is a lower (resp. upper) semi-continuous super-solution (resp. sub-solution) of \textbf{(bHJB)} on $[0,T)\times \mathbb R^m\times \mathbb R\times \mathbb R$ if for all functions $\phi\in \mathcal C^{1,2,2,2}$ and $(\hat t,\hat s,\hat x,\hat y)\in[0,T)\times \mathbb R^m\times \mathbb R\times \mathbb R$ satisfying 
\[0=\min_{(t,s,x,y)\in [0,T)\times \mathbb R^m\times \mathbb R\times \mathbb R} (v-\phi)(t,s,x,y)=(v-\phi)(\hat t,\hat s,\hat x,\hat y),\]
\[\text{resp. } 0=\max_{(t,s,x,y)\in [0,T)\times \mathbb R^m\times \mathbb R\times \mathbb R} (v-\phi)(t,s,x,y)=(v-\phi)(\hat t,\hat s,\hat x,\hat y),\]
    we have
    \[-\partial_t v(\hat t,\hat s,\hat x,\hat y) - (\hat x - \hat y) f(\hat t) -\sup_{(z^x, z, g^x, g, u)}    \mathcal H^v(\hat t,\hat s,\hat x,\hat y, \nabla v, \Delta v, z^x, z, g^x, g, u) \geq 0\; (\text{resp. }\leq 0).\]
    If $v$ is both a super-solution and a sub-solution, we say that $v$ is a viscosity solution to \textbf{(bHJB)}.
\end{definition}
As a consequence of the dynamic programming principle, we have the following theorem
\begin{theorem}\label{thm:visco}
    Assume that the value function $V$ is locally bounded on $[0,T)\times \mathbb R^m\times\mathbb R\times \mathbb R$. Then $V$ is a viscosity solution to \textbf{(bHJB)}.
\end{theorem}
\begin{remark}
    Note that the infinitesimal generator $\mathcal H$ contains a degenerate term through the $\lambda$ process. This term explodes when $t$ approaches $T$ in the case that the support of $\tau$ is bounded and requires to prove Theorem \ref{thm:visco} with a localization technique for both the state variables of the problem and the $\lambda$ process.
\end{remark}

\begin{proof}[Proof of Theorem \ref{thm:visco}] We denote by $V_*,V^*$ the lower semi-continuous and upper semi-continuous envelopes of $V$ respectively.\\

    \textit{Step 1. Proof of the super-solution property.} Let $\phi$ be a $\mathcal C^{1,2,2,2}$ function. Let $(\hat t,\hat s,\hat x,\hat y)\in [0,T)\times \mathbb R^m\times \mathbb R\times \mathbb R$ be such that 
    \[0=\min (V_*-\phi)(t,s,x,y)=(V_*-\phi)(\hat t,\hat s,\hat x,\hat y)\] and a sequence $(t_n,s_n,x_n,y_n)\in [0,T]\times \mathbb R^m\times \mathbb R\times \mathbb R$ such that 
    \[ (t_n,s_n,x_n,y_n)\longrightarrow  (\hat t,\hat s,\hat x,\hat y),\] with 
    $\lim\limits V(t_n,s_n,x_n,y_n) =V_*(\hat t,\hat s,\hat x,\hat y)$. We define $\varepsilon_n:= V(t_n,s_n,x_n,y_n)-\phi(t_n,s_n,x_n,y_n)\geq 0$. Note that $\varepsilon_n\longrightarrow 0$ for large $n$. We also define $X^n:=X^{t_n,x_n,\pi^*},Y^n:=Y^{t_n,y_n,\pi^*}$ the solution to (SDE) starting at time $t_n$ with respective values $x_n$ and $y_n$ and controlled by $\alpha\in \mathcal U$ such that $\alpha=\nu\in \mathcal U_{\hat t}$, so that $\nu$ is a constant at time $t=\hat t$, and $S^n$ denotes the price process starting at the price vector $s_n$ at time $t_n$. In other words, $S^n_{t_n}=s_n, X^n_{t_n}=x_n$ and $Y^n_{t_n}=y_n$. We also define 
    \[\delta_n:= \sqrt{\varepsilon_n}\mathbf 1_{\varepsilon_n\neq 0} + \frac1n \mathbf 1_{\varepsilon_n= 0},\]
and the stopping time

\[\theta_n=\inf\{t>t_n,\; (t,S^n_t,X^n_t,Y^n_t,\lambda_t)\notin [0,t_n+\delta_n)\times \mathcal B_n\times [0,n]\},\]
where $\mathcal B_n$ is the unit ball on $\mathbb R^m\times\mathbb R\times \mathbb R$ centered at the point $(s_n,x_n,y_n)$. We note that $\theta_n\longrightarrow \hat t$ when $n$ goes to $\infty$. According to the dynamic programming principle \eqref{DPP} we have 

\[ V(t_n,s_n,x_n,y_n)\geq  \E\bigl[ V(\theta_n, S_{\theta_n}^{n},X_{\theta_n}^{n},Y_{\theta_n}^{n}) + \int_{t_n}^{\theta_n}U^P(X_{r}^{n} - Y^{n}_r) f(r)dr \bigl]
\]
or equivalently

\begin{align*}
0&\leq  \E\bigl[V(t_n,s_n,x_n,y_n)- V(\theta_n, S_{\theta_n}^{n},X_{\theta_n}^{n},Y_{\theta_n}^{n}) - \int_{t_n}^{\theta_n}U^P(X_{r}^{n} - Y^{n}_r) f(r)dr \bigl]\\
&\leq \varepsilon_n+ \E\bigl[\phi(t_n,s_n,x_n,y_n)- \phi(\theta_n, S_{\theta_n}^{n},X_{\theta_n}^{n},Y_{\theta_n}^{n}) - \int_{t_n}^{\theta_n}U^P(X_{r}^{n} - Y^{n}_r) f(r)dr \bigl]
\end{align*}

Applying Ito's formula for the function $\phi$ and by the localization procedure before the stopping time $\theta_n$ we get
\begin{align*}
0&\leq \varepsilon_n+ \E\bigl[\int_{t_n}^{\theta_n}\big( -\partial_t \phi(r,S_r^n,X_r^n,Y_r^n)- \mathcal H^\phi(r,S_r^n,X_r^n,Y_r^n,\nu)-U^P(X_{r}^{n} - Y^{n}_r) f(r)\big)dr \bigl]\\
&\leq \frac{\varepsilon_n}{\delta_n}+ \E\bigl[\frac1{\delta_n}\int_{t_n}^{\theta_n}\big( -\partial_t \phi(r,S_r^n,X_r^n,Y_r^n)- \mathcal H^\phi(r,S_r^n,X_r^n,Y_r^n,\nu)-(X_{r}^{n} - Y^{n}_r) f(r)\big)dr \bigl]
\end{align*}
Since the function $\phi$ is assumed to be $\mathcal C^{1,2,2,2}$, $\phi$ and its derivative are essentially locally bounded around $\mathcal B_n$ and uniformly in $n$ for $t<T$. Taking the limit as $n$ goes to $\infty$, we obtain 

\[-\partial_t \phi(\hat t,\hat s,\hat x,\hat y)- \mathcal H^\phi(\hat t,\hat s,\hat x,\hat y,\nu)-(\hat x - \hat y) f(\hat t)\geq 0,\text{ for any }\nu\in \mathbb R^{m}\times\mathbb R\times \mathbb R^m\times \mathbb R.\]
We deduce that $V$ is a viscosity super-solution in the sense of Definition \ref{viscosity:def} to \textbf{(bHJB)}.\\

        \textit{Step 2. Proof of the sub-solution property.} Let $\phi$ be a $\mathcal C^{1,2,2,2}$ function. Let $(\hat t,\hat s,\hat x,\hat y)\in [0,T)\times \mathbb R^m\times \mathbb R\times \mathbb R$ be such that 
    \[0=\max (V^*-\phi)(t,s,x,y)=(V^*-\phi)(\hat t,\hat s,\hat x,\hat y).\] Assume by contradiction that 
    \begin{equation*}\label{contradiction} -\partial_t \phi(\hat t,\hat s,\hat x,\hat y)-\sup_\nu \mathcal H^\phi(\hat t,\hat s,\hat x,\hat y,\nu)-(\hat x - \hat y) f(\hat t)>0\end{equation*}
    Since the function $\mathcal H^\phi$ is continuous, we deduce that there exists a ball $\mathcal B_\varepsilon$ around $(\hat t,\hat s,\hat x,\hat y)$ with norm $\varepsilon>0$ small enough such that 
    \[ -\partial_t \phi(t,s,x,y)- \sup_\nu \mathcal H^\phi(t,s,x,y,\nu)-(x-y) f(t)<0,\quad (t,s,x,y)\in \mathcal B_\varepsilon.\]
    Note that there exists $\eta>0$ independent of $\nu$ such that \[-2\eta:=\max_{(t,s,x,y)\in \mathcal B_\varepsilon}V^*(t,s,x,y)-\phi(t,s,x,y).\]
    Let $(t_n,s_n,x_n,y_n)\in \mathcal B_\varepsilon$ converge to $(\hat t,\hat s,\hat x,\hat y)$ such that $\lim\limits V(t_n,s_n,x_n,y_n)=V^*(\hat t,\hat s,\hat x,\hat y)$ and $-\eta<(V-\phi)(t_n,s_n,x_n,y_n)<0$ for any $n\geq 1$. We define 
    \[\theta_n=\inf\{t>t_n,\; (t,S^n_t,X^n_t,Y^n_t,\lambda_t)\notin \mathcal B_\varepsilon\times [0,n]\}.\]
Applying Ito's formula, we get for any control $\nu$
\begin{align*}
    V(t_n,s_n,x_n,y_n)&\geq -\eta+\phi(t_n,s_n,x_n,y_n)\\
    &=-\eta +\mathbb E\big[\phi(\theta_n,S_{\theta^n}^n,X_{\theta_n}^n,Y_{\theta^n}^n) -\int_{t_n}^{\theta_n}[\partial_t \phi(r,S_r^n,X_r^n,Y_r^n)+\mathcal H^{\phi}(r,S_r^n,X_r^n,Y_r^n,\nu_r)] dr\big]\\
    &\geq -\eta +\mathbb E\big[\phi(\theta_n,S_{\theta^n}^n,X_{\theta_n}^n,Y_{\theta^n}^n) +\int_{t_n}^{\theta_n}U^P(X_r^n-Y_r^n)f(r) dr\big]\\
    &\geq \eta+\mathbb E\big[V^*(\theta_n,S_{\theta^n}^n,X_{\theta_n}^n,Y_{\theta^n}^n) +\int_{t_n}^{\theta_n}U^P(X_r^n-Y_r^n)f(r) dr\big].
\end{align*}
 Since $\eta$ is independent of the control $\nu$, it contradicts the dynamic programming principle \eqref{DPP}. We thus deduce that 
 \[ -\partial_t \phi(\hat t,\hat s,\hat x,\hat y)-\sup_\nu \mathcal H^\phi(\hat t,\hat s,\hat x,\hat y,\nu)-(\hat x - \hat y) f(\hat t)\leq 0,\]
 hence $V$ is a sub-solution in the sense of Definition \ref{viscosity:def} to \textbf{(bHJB)}.
\end{proof}

\subsubsection{Unbounded default time}
Even in the unbounded case, we start by getting rid of the $\hat{Y}_0$ we have in \eqref{eq:preHJB} as it is useless for the optimization routine: again $\hat{Y}$ and $Y$ only differ by a constant, so they have the same governing SDE. Therefore, our starting point is the following optimization problem
\[
    \sup_{(Z, Z^X, \Gamma, \Gamma^X, U)} \mathbb{E}[U^P(X_{\T}^{\pi^*} - \hat Y^{\pi^*}_{\T})]\\
\]

The first difference is that now, in reformulating the problem using the default density, we will still have a terminal part $(1 - F_{\tau}(T))(X_T - Y_T)$. Overall, our problem is now 
\begin{equation*}
    \E\bigl[ \int_{0}^T U^P(X_t^{\pi^*} -  \hat Y^{\pi^*}_t) f(t)dt + (1 - F_{\tau}(T))U^P(X_T -  \hat Y^{\pi^*}_T)\bigl]
\end{equation*}

Note that this problem differs from the bounded default time case only from the terminal condition given by
    \[
        v(T, x, y) = (1 - F_{\tau}(T))U^P(x - y)
    \]

We introduce the following integro-partial PDE and the verification theorem follows.

\[
\textbf{(uHJB)}\quad\begin{cases}
   \partial_t v(t,s,x,y) + U^P(x - y) f(t) + \sup_{(z^x, z, g^x, g, u)}\bigl \{   \mathcal H^v(t,s, x, y, \nabla v, \Delta v, z^x, z, g^x, g, u) \bigl \}=0,\; t<T\\
v(T,s,x,y)= (1 - F_{\tau}(T))U^P(x - y).
\end{cases}
\]

\begin{theorem}[Verification Theorem - unbounded case]
Assume that there exists a function $\phi$ twice continuously differentiable in space and differentiable in time, such that $\phi(t,s,x,y)$ solves  \textbf{(uHJB)}. We denote by $(\hat Z_t, \hat Z_t^X, \hat \Gamma_t, \hat \Gamma_t^X, \hat U_t)$ the optimizers in the supremum. 
Furthermore, assume that $\phi$ has a quadratic growth in $y$ and polynomial growth in $s,x$ such that 
\[|\phi(t,s,x,y)|\leq \kappa(1+|x|^p+\|s\|^p+|y|^2),\; p>1,\; \kappa>0.\]

Then, for each $t \in [0,T]$, the strategy $(\hat{Z}^X, \hat{Z}, \hat{\Gamma}^X, \hat{\Gamma}, \hat{U})$ is an optimal strategy for the control problem and 

\begin{align*}
    \phi(0,S_0,x,0)=V_0=  \sup_{(Z^X, Z, \Gamma^X, \Gamma, U)} \mathbb{E}\left[X_{T \wedge \tau}^{\pi^*} -  \hat Y^{\pi^*}_{T \wedge \tau}\right].
\end{align*}

The optimal contract is given by 
\begin{align*}
    \xi^\star 
    &= \hat R_0+ \int_{0}^{\T} \sum_{i=1}^m \frac{\hat{Z}^i_t}{S^i_t} dS^i_t + \int_0^{\T} \hat{Z}_t^X dX_t + \int_{0}^{\T} \hat{U}_s dH_s \\
    &+ \frac{1}{2}\int_0^{\T} (\hat{\Gamma}_t^X + \eta (\hat{Z}_t^X)^2 ) d\langle X, X\rangle_t - \int_0^{\T} F(t, \hat{Z}_t, \hat{Z}_t^X, \hat{\Gamma}_t, \hat{\Gamma}_t^X,U_t) dt\\
    &+\int_0^{\T} \sum_{i=1}^m \frac{\hat{\Gamma}^i_s}{S^i_t} d\langle S^i, X\rangle_t 
\end{align*}
\end{theorem}

\begin{remark}
    The proof of this theorem follows the same lines as the proof of Theorem \ref{thm:verif} without requiring the localization of the $\lambda$ term.
\end{remark}
\begin{remark}
    Similarly to Theorem \ref{thm:visco}, the value function of the problem is a viscosity solution to \textbf{(uHJB)}.
\end{remark}

\paragraph{A note on the uniqueness and numerical motivation.} In order to carry out numerical simulations, we must first address the uniqueness of viscosity solutions to the integro-partial differential equations \textbf{(bHJB)} and \textbf{(uHJB)}. In the integrable case, \textit{i.e.,} when $[0,T]$ is strictly contained in the support of $\tau$ so that $\lambda$ is integrable on $[0,T]$, uniqueness follows directly from existing results; see, for instance, \cite[Section~5]{seydel2009existence}. The bounded case, in which the support of $\tau$ is contained in $[0,T]$, is considerably more delicate. The loss of integrability of $\lambda$ as $t$ approaches $T$ introduces a degeneracy that, to the best of our knowledge, is not treated in the existing literature. We outline below a possible strategy to resolve this difficulty, but we leave a full analysis for future research, as it would require deep results in the theory of second-order BSDEs. Observe that the probabilistic formulation of \textbf{(bHJB)} corresponds to a second-order BSDE with jumps (see, for example, \cite[Section~3.4]{gennaro20252bsde}), but with a degenerate (singular) generator, similar to the one studied in \cite{jeanblanc2014note}. A first step would be to extend the analysis of \cite{jeanblanc2014note} to the second-order setting by studying the convergence of the solutions to the second-order BSDEs with truncated intensities
$\lambda^n := \lambda \wedge n,$
and by proving that these solutions converge to the solution of the 2BSDE with the original singular generator. Next, by connecting the solutions of the truncated 2BSDEs to the viscosity solutions of the corresponding non-degenerate IPDEs, a comparison theorem would be obtained for the viscosity solutions of the truncated IPDEs, using, for example, \cite{alvarez1996viscosity,mou2015uniqueness}. Finally, passing to the limit as $n \to \infty$ would yield a comparison theorem for viscosity solutions of the original IPDE with singular coefficients, thereby establishing the uniqueness of its viscosity solution.

\section{Numerical solutions}\label{chap5}

In this section, we study numerical solutions to the coupled HJB equations
\textbf{(bHJB)} and \textbf{(uHJB)}. Our goal is twofold. First, we want to contrast the bounded and unbounded default specifications and understand how they differ. Second, under a bounded default time, we are interested in extracting qualitative features of the optimal contract and trading policy. In particular, we aim to:
\begin{itemize}
    \item understand how the random default time affects both the optimal
    trading strategy and the incentive compensation scheme;
    \item investigate how the shape of the default distribution on $[0,T]$ impacts the optimal compensation;
    \item compare, in the bounded default case, the behaviour of the
    agent under the linear specification $\Xi^l$ and under the general functional $\Xi$;
    \item analyse how the different sources of risk (market risk and default
    risk) translate into distinct patterns for the various incentive
    components.
\end{itemize}
From a numerical point of view, the problem is intrinsically iterative. The maximizers $(\hat \pi,\hat Z^X,\hat \Gamma^X,\hat U)$ entering
\textbf{(bHJB)} and \textbf{(uHJB)} depend on the value function and its
derivatives, and in particular $\hat U$ is characterized as the argmax of a functional that explicitly involves both $v(t,x,y)$ and the shifted value $v(t,x,y+u)$. The coefficients of the PDE therefore depend nonlinearly on the
unknown solution itself, which rules out straightforward one–shot schemes.
To keep the numerical exploration as transparent as possible, we focus on the simpler $1-$dimensional case with a single risky asset. This allows us to isolate and better interpret the
dynamics of the default–linked incentive $U$. In this one–dimensional case the compensation for exposure to the traded asset $S$ is switched off, as it could be seen as a reward for the agent to invest more aggressively on the riskier assets and in the $1-$dimensional case it becomes essentially the same as incentivize on the portfolio growth. Therefore, the only diffusion incentive is the process $Z^X$ attached to the wealth process $X$. Both principal and agent are modeled with exponential utility, with possibly different risk–aversion coefficients. An additional objective of our numerical study is to understand how the optimal default–contingent incentive $U$ reacts to changes in the ratio between the principal’s and the agent’s risk aversion.

In the recent literature, neural networks showed great potential to be the state of the art for numerical schemes for PDEs, so the backbone of our algorithm will be a simple feed-forward neural network, that we will deploy in a deep learning scheme. Similar ideas have been previously developed in \cite{baldacci2019market}, where an actor--critic approach is used to find the optimal policy for market making together with the value function involving an integro-partial HJB equation in high dimension and in \cite{lu2025multiagent} in which the authors use reinforcement learning methods to solve the problem. Note that in our case we are using supervised learning for an integro-partial HJB equation which differs from the existing literature on stochastic control problems.

We recall that the value function $v$ solves the HJB equation
\begin{equation}\label{eq:HJB-numeric}
    \partial_t v(t,x,y)
    \;+\;
    U^P(x,y)\, f(t)
    \;+\;
    \sup_{(z^x,\,g^x,\,u)}
    \Big\{
        \mathcal{H}^v\bigl(
            t,x,y,\nabla v(t,x,y),\Delta v(t,x,y),
            z^x,g^x,u
        \bigr)
    \Big\}
    \;=\; 0,
    \qquad t < T,
\end{equation}
The running term $U^P(x,y) f_T(t)$ captures the principal's instantaneous utility weighted by the density $f_T$ of the default time, while the Hamiltonian $\mathcal{H}^v$ contains the drift, diffusion, and default terms, as well as the dependence on the controls $(z^x,g^x,u)$.\\ 
We equip this PDE with a terminal condition $v(T,x,y)=\Phi(x,y)$ and a boundary condition on the solvency boundary $x=0$, which in our specification reads $v(t,0,y) = -1$. The need for a boundary condition is given by the fact that we need to solve a PDE that lives in $[0, T] \times \R \times \R$, but practically we have to solve it in a bounded region. The economic interpretation is that if the portfolio goes to $0$ the principal does not pay the agent anything, creating a fixed value of $-1$. \\
The numerical algorithm proceeds by alternating between two steps. We start from an initial guess for the controls,
\[
    \bigl( {z_0^x}^*, {g_0^x}^*, u_0^* \bigr),
\]
and an initial, untrained neural network representing a value function $v_0$. For this fixed choice of $(z^x,g^x,u)$ we solve the HJB equation \eqref{eq:HJB-numeric} by training the neural network to obtain an updated value function $v_1$. Following the Physics--Informed Neural Networks (PINNs) approach \cite{raissi2017physics}, we use automatic differentiation to compute the first and second derivatives of $v_\theta$, construct a loss function as the squared residual of the PDE together with penalty terms enforcing the terminal and boundary conditions, and update the network parameters $\theta$ via stochastic gradient descent.\\
Once a new approximation $v_1$ is available, we can solve \emph{numerically} the pointwise optimization problem
\[
    \sup_{(z^x,\,g^x,\,u)}
    \Big\{
        \mathcal{H}^v\bigl(
            t,x,y,\nabla v_1(t,x,y),\Delta v_1(t,x,y),
            z^x,g^x,u
        \bigr)
    \Big\},
\]
thereby obtaining updated controls
$\bigl( {z_1^x}^*, {g_1^x}^*, u_1^* \bigr)$.
We then freeze these new controls and solve again the HJB equation via the PINN, producing $v_2$, and we iterate this actor--critic-style procedure until convergence of both the value function and the control processes. A high–level pseudo–code of the resulting alternating scheme is reported in
Algorithm 1 and Algorithm 2.
\vspace{0.5cm}

\begin{algorithm}[H]
\SetKwFunction{SolvePDENN}{SolvePDENN}
\SetKwFunction{SolveMax}{SolveMax}
\SetKwProg{Fn}{Function}{:}{}

\Fn{\SolvePDENN{$s_k, {z_k^x}^*, z^*_k, {g_k^x}^*, g^*_k, u^*_k, v_k$}}{
    \tcp{Step 1: Evaluate the PDE residual at the sampled points $s_k$ using current variables, and evaluate boundary and terminal loss}
    Compute PDE residual at $s_k$ using $v_k, {z_k^x}^*, {g_k^x}^*, u^*_k$\;
    
    \tcp{Step 2: One step of stochastic gradient descent for the neural network-based approximation of the PDE}
    Use a neural network or numerical solver to minimize the residual and update $v_{k+1}$\;
    
    \tcp{Step 3: Return the updated PDE solution}
    \Return $v_{k+1}$\;
}

\Fn{\SolveMax{$v_{k+1}$}}{
    \tcp{Step 1: Compute the current approximation of the value function $v_{k+1}$}
    Plug in $v_{k+1}$ into the function to be maximized\;
    
    \tcp{Step 2: Maximize over the variables}
    Solve for ${z_{k+1}^x}^*, {g_{k+1}^x}^*, u^*_{k+1}$ that maximize the objective function (grid approach)\;
    
    \tcp{Step 3: Return the maximized variables}
    \Return ${z_{k+1}^x}^*,  {g_{k+1}^x}^*, u^*_{k+1}$\;
}
\caption{Subroutines: SolvePDENN and SolveMax}
\end{algorithm}

\begin{algorithm}[H]
   \SetAlgoNlRelativeSize{0}
   \KwData{
      $v_0$: Initial guess for the PDE solution\\
      $({z_0^x}^*, {g_0^x}^*, u^*_0)$: Initial maximization variables\\
      Sampler: sampling of the variables\\
      PDE inputs: parameters of the PDE\\
      Tolerance $\epsilon$, max\_iter: Maximum iterations\\
      $k \leftarrow 0$: Iteration counter
   }
   \KwResult{PDE solution $v^*$}
   
   \While{$k \leq \text{max\_iter}$ \textbf{ and } \text{not converged}}{
      \tcp{Sample new points}
      $s_{k} \leftarrow$ sampler$()$\;
      
      \tcp{Solve the PDE with current guess}
      $v_{k+1} \leftarrow$ \SolvePDENN$(s_k, {z_{k}^x}^*, {g_{k}^x}^*, u^*_k, v_k)$\;
      
      \tcp{Maximize over variables}
      ${z_{k+1}^x}^*, {g_{k+1}^x}^*, u^*_{k+1} \leftarrow$ \SolveMax$(v_{k+1})$\;
      
      \tcp{Check for convergence}
      \If{$|v_{k+1} - v_k| \leq \epsilon$}{
         \tcp{If converged, exit loop}
         break\;
      }
      
      \tcp{Increment iteration counter}
      $k \leftarrow k + 1$\;
   }
   
   \caption{PDE-Maximizer Structure}
\end{algorithm}

\paragraph{Main numerical challenges and remedies.}
The solution of the coupled HJB system \textbf{(bHJB)}-\textbf{(uHJB)} presents several numerical difficulties, both at the level of the PDE and in the iterative optimization of the controls. A first challenge is that the neural network may converge prematurely to an almost constant $0$ function, given the nature of the exponential utility function: this is alleviated by the boundary condition and some pre-training on this condition to bring the initial solution away from $0$.
A second difficulty lies in the relative scaling of the PDE residual, boundary, and terminal losses. To mitigate this effect we employ self-aware weights: the loss has $3$ components, whose weights depend on the relative magnitude of the components, so that we inform the network which part to improve the most. \\
A third source of complexity arises from the iterative structure of the
problem. The coefficients of the Hamiltonian depend on the maximizers
$(z^{x},\Gamma^{x},u)$, which in turn depend on the value function and its derivatives. This creates a fixed-point structure: given a value function estimate $v$, one must solve three maximization problems to obtain updated controls; these controls define a new PDE, whose solution yields a new $v$, and the process repeats. In practice, this nested optimization may lead to
oscillations or loss explosions, especially early in the training process, especially when the exponential utilities amplify
errors or when the default intensity becomes large. To alleviate this
instability we update the controls through a smoothed scheme: instead of
using the new maximizers directly, we form an exponential moving average with
the previous controls. This damping mechanism prevents violent jumps in the
coefficients and results in a smoother iterative trajectory. The smoothing effect vanishes over time as we start trusting more and more the output of the optimization problems.\\
Regarding the optimization over $(z^{x},\Gamma^{x},u)$, we rely primarily on
grid search, which proved robust in the early stages of training when the value function is still crude and derivative estimates are noisy. More sophisticated gradient--based maximization was found to be less reliable, as it tended to produce non--smooth or unstable controls in the first iterations. Once the
controls are identified on the residual grid, a $k$--nearest--neighbour interpolation is applied during the forward simulation of the BSDE to recover pathwise controls at arbitrary states.

In particular, we deployed a simple feed-forward neural network (NN), with $3$-d inputs, $1$-d outputs, $8$ hidden layers each with $32$ neurons. The learning rate was piecewise linear, we were sampling $8192$ points uniformly and $2048$ between terminal and boundary points. Our first objective is to understand how the \emph{shape} of the default distribution influences the optimal control structure. Since skewness plays a particularly important role in determining when default is more likely to occur, we focus on the Beta$(a,b)$ family on $[0,T]$, which allows us to adjust the parameters $(a,b)$ to generate different skewness profiles. In this way we can systematically test how early-default–heavy versus late-default–heavy scenarios affect the dynamic incentives encoded in $U$, $Z^X$, and $\Gamma^X$. We worked with Beta$(2, 4)$, Beta$(1, 1)$ (i.e., a uniform distribution), and Beta$(4, 2)$. Of course, the exploding compensator $\lambda$ was computed but capped. The market has a positive drift of $10\%$ with a volatility of $30\%$. In particular, we will show as an example the calculations for the symmetric beta random variable, which is, in fact, a uniform distribution. Consider, for example, $\tau$ being uniformly distributed on $[0,1]$. We recall that
\[
    \mathbb P(\tau > x | \mathcal{F}_t)= \int_x^{+\infty} \gamma(t, u) du.
\]
Since we assume independence of $\tau$ from the filtration, we have \[\mathbb P(\tau > x | \mathcal{F}_t) = \mathbb P(\tau > x) = 1-x\]
as $\tau \sim U(0,1)$. Consequently, 
\[
    \gamma(t, u) =  \ind_{1 \geq u},\; \lambda_t = \frac{\gamma(t, t)}{\mathbb P(\tau > t | \mathcal{F}_t)} = \frac{1}{1-t}.
\]

Using \cite[Proposition 4.4]{el2010happens}, by defining
\[
    \Lambda_t = \int_0^t \lambda_s ds,
\]
we have 
\[
    \mathbb P(\tau > x | \mathcal{F}_t) = \exp{(-\Lambda_t)}.
\]
Therefore
\begin{align*}
    \Lambda_t = \int_0^t \lambda_s ds = \int_0^t \frac{1}{1-s}ds= -\ln{(1-t)},
\end{align*}
so that
\[
    \exp{(-\Lambda_t)} = \exp{(\ln{(1-t)})} = 1 -t,
\] matching the survival function of the uniform distribution. 
In all the experiments with an independent default time, $\tau$ has a probability density function $f_{\tau}$ and an associated cumulative distribution function $F_{\tau}$ so that
\[
    \lambda_t = \frac{f_{\tau}(t)}{1 - F_{\tau}(t)}.
\]

\paragraph{Comparison of exponential and uniform random default} We now undertake a comparative analysis of trading policies and incentive structures under bounded versus unbounded default horizons on $[0,1]$. For the unbounded setting, we specify the default time as exponentially distributed with intensity $\lambda = 1.5$, and contrast it with the uniform distribution on $[0,1]$ for the bounded benchmark. Figure~\ref{fig:exponential_vs_uniform} displays, for both default specifications, the resulting portfolio trajectories, optimal investment strategies, and incentive processes related to the default $U$ and the average incentive $Z^X$ associated with the portfolio dynamics. A salient feature emerges from the finite-horizon structure: under the exponential (unbounded) default, the optimal trading policy is initially more aggressive. This reflects the relatively high short-term hazard rate, which rationalizes front-loaded risk-taking in order to extract value before a potentially early default.
The optimal policy exhibits a transition pattern: initially, investment levels are higher under the exponential default relative to the uniform case; subsequently, the allocation becomes more conservative for the unbounded case; and, finally, it declines more sharply for the bounded case as the terminal date approaches.
Near maturity, the optimal strategy becomes approximately linear, consistent with standard finite-horizon contraction effects. Intuitively, the exponential default concentrates probability mass near the origin, inducing an initially aggressive stance followed by a gradual reduction in exposure, whereas the uniform default, exhibiting a flat hazard rate, generates a smoother, more homogeneously declining investment schedule.
Regarding incentives, both default-driven and portfolio-value–based compensations are systematically higher in the bounded case. This is economically consistent: when default is guaranteed to occur before $T$, the principal must provide stronger incentives to align the risk-averse agent’s actions with the desired investment profile. Conversely, $Z^X$ is naturally decreasing as $t \to T$ in the bounded scenario, reflecting the rising conditional probability of default and the diminishing marginal need to incentivize effort.
\begin{figure}[h]
    \centering

    \begin{subfigure}{0.45\textwidth}
        \centering
        \includegraphics[width=\linewidth]{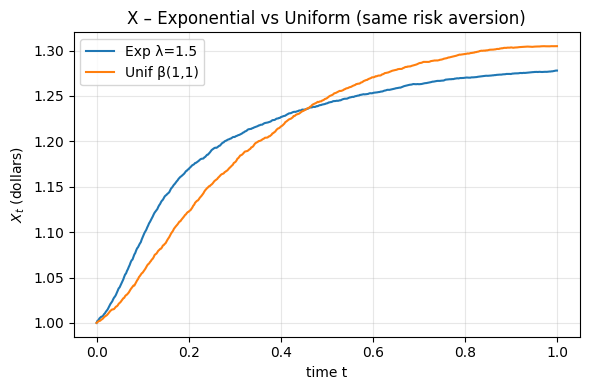}
        \caption{Portfolio value}
    \end{subfigure}
    \hfill
    \begin{subfigure}{0.45\textwidth}
        \centering
        \includegraphics[width=\linewidth]{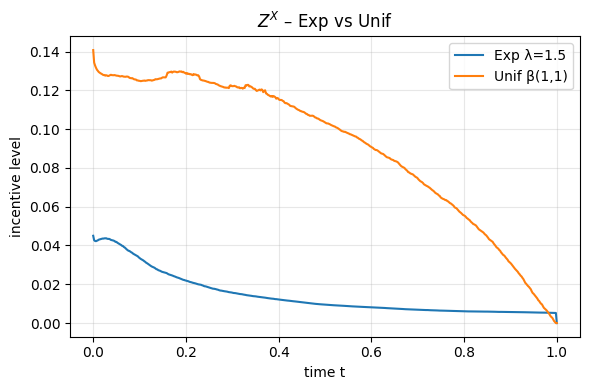}
        \caption{Average $Z^{X}$}
    \end{subfigure}

    \vspace{0.5cm}

    \begin{subfigure}{0.45\textwidth}
        \centering
        \includegraphics[width=\linewidth]{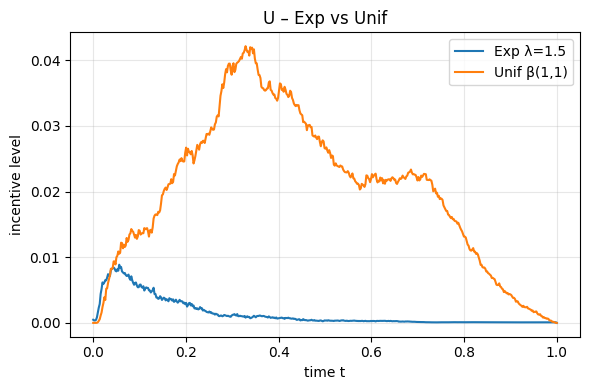}
        \caption{Average $U$}
    \end{subfigure}
    \hfill
    \begin{subfigure}{0.45\textwidth}
        \centering
        \includegraphics[width=\linewidth]{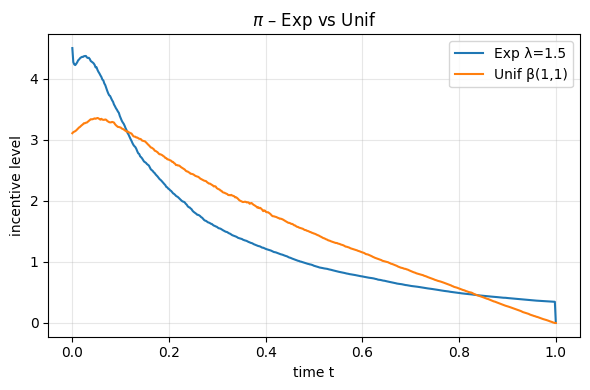}
        \caption{Optimal strategy}
    \end{subfigure}

    \caption{Random default with uniform distribution on $[0,T]$ and exponential distribution. Top left: portfolio value for bounded and unbounded distributions. Top right: average $Z^{X}$ for bounded and unbounded distributions. Bottom left: average $U$ for bounded and unbounded distributions. Bottom right: optimal investment strategy for bounded and unbounded distributions.}
    \label{fig:exponential_vs_uniform}
\end{figure}

\paragraph{Sensitivities with respect to the risk aversion}
In Figure~\ref{fig:sensitivityriskaverse}, we investigate the impact of the ratio $\frac{\eta}{\varsigma}$ under a uniform default time on $[0,T]$ with $T=1$ on the portfolio value, as well as on the compensations $Z^X$ and $U$ with respect to the default time and the optimal investment strategy. We first observe that when the agent and the principal share the same risk-aversion parameter, the portfolio value is higher for lower levels of risk-aversion as expected in risky investments. Alternatively, if the agent is more risk-averse than the principal, the resulting portfolio value decreases. In other words, delegating the portfolio to a more risk-averse manager is less efficient, and alignment of risk-aversion parameters leads to improved portfolio performance. This phenomenon is explained by the optimal investment strategy: a manager who is less risk-averse than the client invests more aggressively in the market, generating a form of risk-transfer effect. This is reflected in a higher compensation $Z^{X}$ (blue versus red curves), and, more interestingly, in a higher compensation $U$ associated with the random default risk.

\begin{figure}[h]
    \centering

    \begin{subfigure}{0.45\textwidth}
        \centering
        \includegraphics[width=\linewidth]{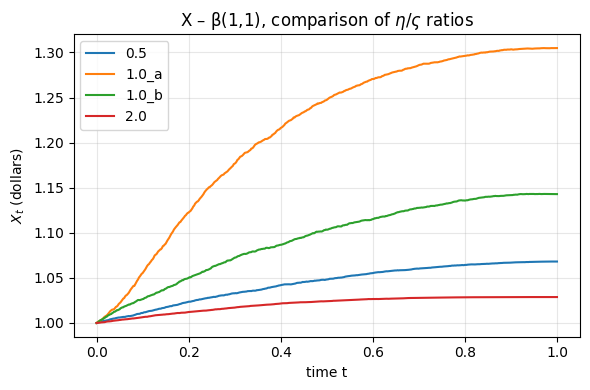}
        \caption{Portfolio value}
    \end{subfigure}
    \hfill
    \begin{subfigure}{0.45\textwidth}
        \centering
        \includegraphics[width=\linewidth]{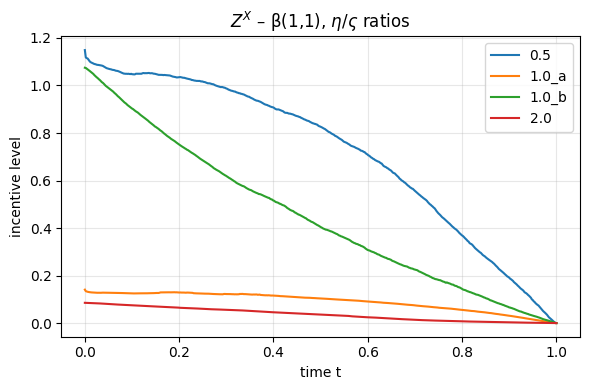}
        \caption{Average $Z^{X}$}
    \end{subfigure}

    \vspace{0.5cm}

    \begin{subfigure}{0.45\textwidth}
        \centering
        \includegraphics[width=\linewidth]{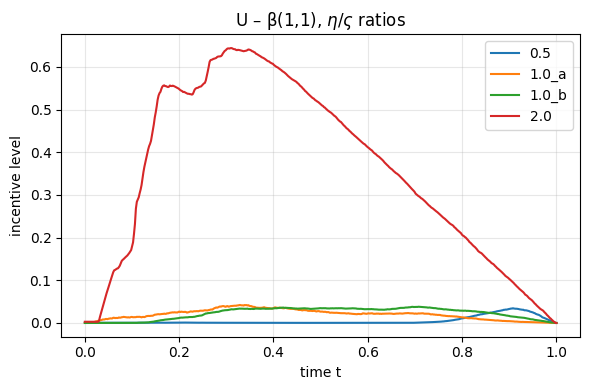}
        \caption{Average $U$}
    \end{subfigure}
    \hfill
    \begin{subfigure}{0.45\textwidth}
        \centering
        \includegraphics[width=\linewidth]{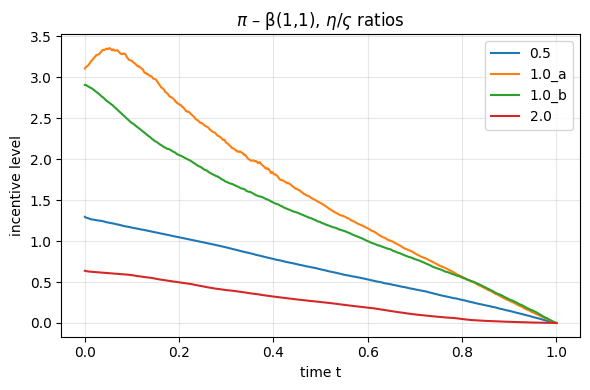}
        \caption{Optimal strategy}
    \end{subfigure}

    \caption{Random default with uniform distribution on $[0,T]$. Top left: average portfolio value for different risk-aversion ratios. 
    Top right: average $Z^{X}$ for different risk-aversion ratios. 
    Bottom left: average $U$ for different risk-aversion ratios. Bottom right: average optimal investment strategy for different risk-aversion ratios. 
    The blue curve corresponds to $\eta = 0.5\,\varsigma$ (agent less risk averse than the principal), 
    the red curve corresponds to $\eta = 2\,\varsigma$ (agent more risk averse than the principal), 
    the orange curve corresponds to $\eta = \varsigma$ with a low level of risk aversion, 
    and the green curve corresponds to $\eta = \varsigma$ with a high level of risk aversion.}
    
    \label{fig:sensitivityriskaverse}
\end{figure}

\paragraph{Impact of the skewness of the distribution}
A very interesting comparison, for the bounded distribution, is the impact of skewness, a good measure for understanding the shape of the distribution and where the probability mass is concentrated. We fix the risk aversion for both the agent and the principal to the higher value of $1$ and obtain the plots in figure \ref{fig:skew}
\begin{figure}[!ht]
    \centering

    \begin{subfigure}{0.45\textwidth}
        \centering
        \includegraphics[width=\linewidth]{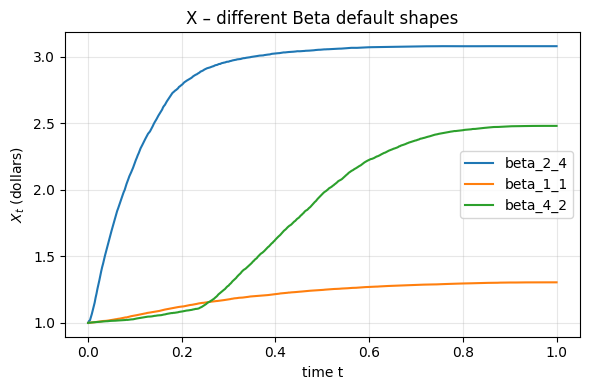}
        \caption{Portfolio value}
    \end{subfigure}
    \hfill
    \begin{subfigure}{0.45\textwidth}
        \centering
        \includegraphics[width=\linewidth]{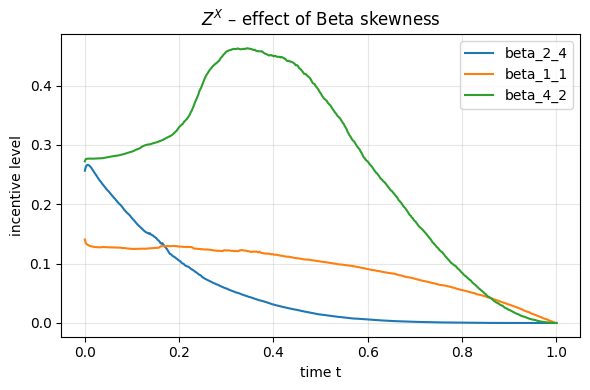}
        \caption{Average $Z^{X}$}
    \end{subfigure}

    \vspace{0.5cm}

    \begin{subfigure}{0.45\textwidth}
        \centering
        \includegraphics[width=\linewidth]{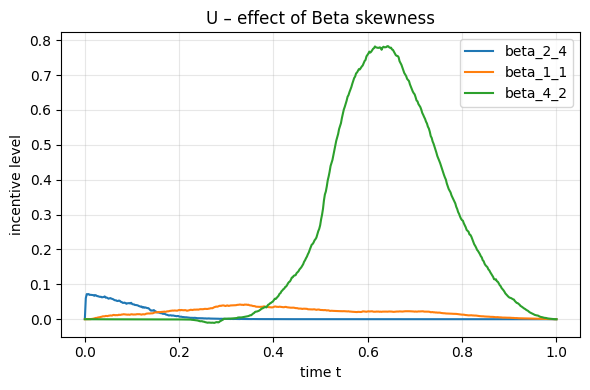}
        \caption{Average $U$}
    \end{subfigure}
    \hfill
    \begin{subfigure}{0.45\textwidth}
        \centering
        \includegraphics[width=\linewidth]{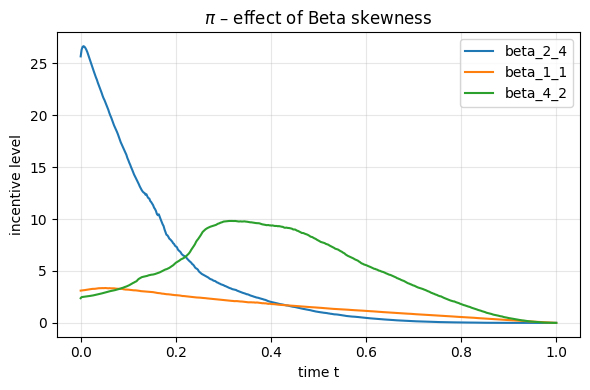}
        \caption{Optimal strategy}
    \end{subfigure}

    \caption{Random default with $3$ different beta distributions on $[0,T]$: in green, we have a left-skewed $\mathrm{Beta}(4,2)$, in blue a right-skewed $\mathrm{Beta}(2,4)$, and in orange a uniform distribution.}
    
    \label{fig:skew}
\end{figure}
Now, it is not surprising that, under different distributions, we can get radically different strategies and portfolio values, as the exponential utility will shrink the differences. The interesting behavior is what happens for the incentives. Let's start by analyzing the right-skewed beta (in blue): the probability mass is all concentrated in the first part of the time interval, therefore the agent knows that most of the compensation will come from the percentage on the wealth and he takes an extremely aggressive strategy. Knowing this, the principal offers a very high share of the portfolio’s growth together with immediate insurance against default. Because the default is expected to occur early, there is little incentive toward the latter part of the time interval, and all quantities converge to zero.

The situation changes in the opposite case of a left-skewed Beta distribution. Here, both the agent and the principal understand that default is more likely to occur in the second half of the horizon. As a result, both have an interest in adopting aggressive strategies; however, for the same level of compensation, the agent feels less pressure to invest quickly, since he expects more time before default. Consequently, the principal must continually increase the incentive rate to motivate the agent to intensify his investment activity.

This generates a delayed response relative to what one might initially expect: the agent becomes more cautious and reduces investment just as the probability of default sharply increases, requiring the principal to offer substantially higher insurance. The reason the insurance component becomes so large is intuitive: on average, the portfolio value tends to be high due to the earlier aggressive strategy, so a late default represents a much greater lost opportunity than a default occurring near the start of the horizon.

\paragraph{Comparison of linear and general contracts}
Figure \ref{fig:linear} shows the difference between the optimal solution when contracts are restricted to the class of linear contracts $\Xi^{l}$ and the solution under the full class of admissible contracts $\Xi$. The incentive rates of 5\% and 20\% are selected since 5\% is the best fixed percentage within the restricted class, while 20\% reflects a common benchmark in the hedge fund industry. The resulting optimality gap in portfolio performance is significant (top left), and in both cases the linear strategies are noticeably more conservative (bottom right). The interaction between the growth-based incentive, the default insurance, and the investment strategy is particularly revealing. When the agent receives only 5\% of the portfolio’s growth, he adopts a more aggressive market stance; however, the risk of default becomes relevant earlier, since he requires sufficient time to expand the portfolio in order to earn a meaningful absolute compensation. Anticipating this, the principal provides earlier default insurance.

At the opposite end, a 20\% incentive, higher than the growth share in the fully optimal contract at the initial time, induces greater risk aversion in the agent, leading to a less aggressive exposure to the risky asset. As a result, the insurance component activates later. Under the optimal growth incentive, the timing of default insurance lies between these two extremes, striking a balance between the competing incentives and supporting a more aggressive strategy when market conditions are favorable.
\begin{figure}[!ht]
    \centering

    \begin{subfigure}{0.45\textwidth}
        \centering
        \includegraphics[width=\linewidth]{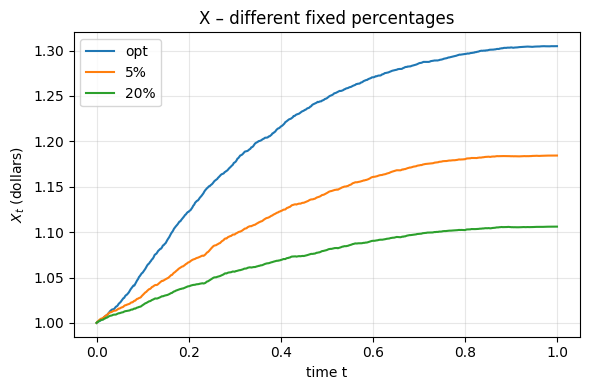}
        \caption{Portfolio value}
    \end{subfigure}
    \hfill
    \begin{subfigure}{0.45\textwidth}
        \centering
        \includegraphics[width=\linewidth]{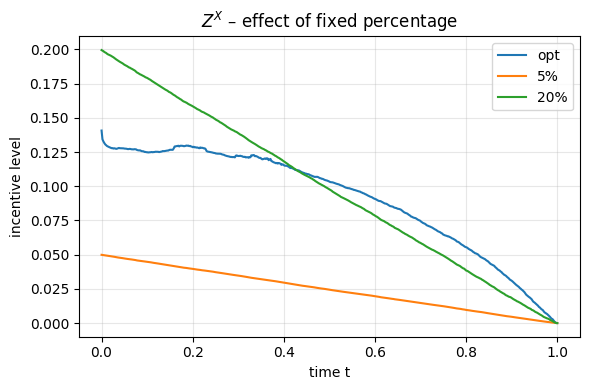}
        \caption{Average $Z^{X}$}
    \end{subfigure}

    \vspace{0.5cm}

    \begin{subfigure}{0.45\textwidth}
        \centering
        \includegraphics[width=\linewidth]{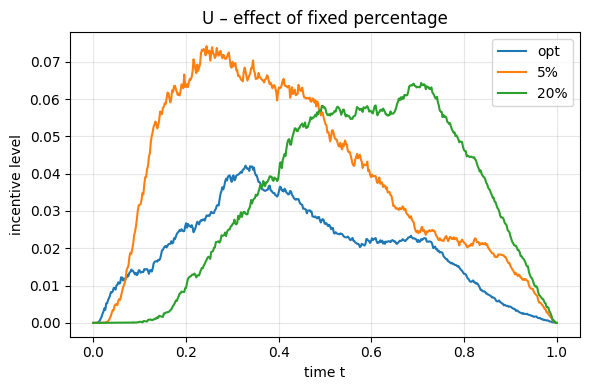}
        \caption{Average $U$}
    \end{subfigure}
    \hfill
    \begin{subfigure}{0.45\textwidth}
        \centering
        \includegraphics[width=\linewidth]{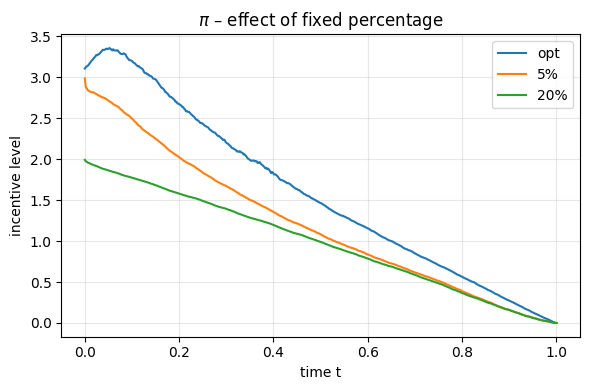}
        \caption{Optimal strategy}
    \end{subfigure}

    \caption{Linear incentive on portfolio growth under uniform default. Top left: portfolio value for different percentages of the incentive. Top right: average $Z^{X}$ for different percentages of the incentive. Bottom left: average $U$ for different percentages of the incentive. Bottom right: optimal investment strategy for different percentages of the incentive.}
    
    \label{fig:linear}
\end{figure}

\newpage

 \bibliographystyle{plain}
 \small
\bibliography{sn-bibliography}

\newpage
\appendix

\section{First-best optimal delegated contract}\label{FB-appendix}
As a benchmark case, we consider here the situation where the principal proposes both a contract and an investment strategy to the agent. We rewrite the principal problem in the first-best case as 
\[V^{FB}_0=\sup_{(\xi,\pi)} \mathbb E\Big[U^P(X^\pi_{T\wedge \tau}-\xi)\Big],\]
subject to 
\[J^A(\pi;x,\xi)\geq R.\]
We consider exponential utilities for both the agent and the principal so that $U^P(x)=-e^{-\varsigma x}$ and $U^A(x)=-e^{-\eta x}$, for risk-aversion parameters $\varsigma,\eta>0$ and no penalty on the trading strategy $\varepsilon=0$. 
Introducing a Lagrange multiplier $\rho\geq 0$, the problem is reduced to

\begin{align*}
V^{FB}_0&=\inf_{\rho\geq 0}\sup_{\pi}\sup_{\xi} \mathbb E\Big[U^P(X^\pi_{T\wedge \tau}-\xi)+\rho (J^A(\pi;x,\xi)-R) \Big],\\
&=\inf_{\rho\geq 0}\sup_{\pi}\sup_{\xi} \mathbb E\Big[U^P(X^\pi_{T\wedge \tau}-\xi)-\rho e^{-\gamma \xi}  -\rho R \Big].
\end{align*}
Using G\^ateaux derivatives, see for example \cite{mastrolia2018moral}, the first-order condition on an optimizer $\xi^\star$ gives
\begin{equation}\label{FOC}e^{-\varsigma (X_{T\wedge \tau}-\xi^\star)}=\frac{\rho \eta}\varsigma e^{-\gamma \xi^\star}.\end{equation}
Consequently,
\[\xi^\star= \frac1{\varsigma+\eta} [\log(\frac{\eta\rho}\varsigma)+\varsigma X_{T\wedge \tau}].\]
Note that the optimal contract is linear with respect to $X_{T\wedge \tau}$ and similar to the contract in \cite[Section 2.1]{cvitanic2017moral} as a function of the Lagrangian $\rho$. Substituting this optimizer into the investor's first-best problem and using the first-order condition \eqref{FOC} we get the following:
\begin{align*}
V^{FB}_0&=\inf_{\rho\geq 0}\sup_{\pi} \mathbb E\Big[U^P(X^\pi_{T\wedge \tau}-\xi^\star)+\rho (J^A(\pi;x,\xi^\star)-R) \Big],\\
&=\inf_{\rho\geq 0}\sup_{\pi} \mathbb E\Big[ -\rho(\frac{\varsigma+\eta}{\varsigma})e^{-\eta \xi^\star}\Big] -\rho R\\
&=\inf_{\rho\geq 0}\sup_{\pi} \mathbb E\Big[ -\rho(\frac{\varsigma+\eta}{\varsigma})   (\frac{\varsigma}{\eta\rho})^{\frac\eta{\varsigma+\eta}}e^{-\frac{\varsigma\eta}{\varsigma+\eta}X_{T\wedge \tau}}\Big] -\rho R\\
&=\inf_{\rho\geq 0} \big\{ \rho(\frac{\varsigma+\eta}{\varsigma})   (\frac{\varsigma}{\eta\rho})^{\frac\eta{\varsigma+\eta}}\tilde V_0^{FB} -\rho R\Big\},
\end{align*}
where 
\[\tilde V_0^{FB}=\sup_\pi \mathbb E\Big[-e^{-\frac{\varsigma\eta}{\varsigma+\eta} X_{T\wedge \tau}}\Big].\]

We introduce the following HJB equations depending on whether the support of $\tau$ is bounded in $[0,T]$ or not.

\[
\textbf{(bHJB-FB)}\quad\begin{cases}
  &  \partial_t \tilde v(t,s,x) + \tilde U(x) f(t) + \sup_{\nu}\bigl \{ \tilde{\mathcal {H}}^{\tilde v}(t,s,x;\nu)\bigl \}=0,\; t<T\\
&\tilde v(T,s,x)=0,\; (s,x)\in \mathbb R^m\times \mathbb R,
\end{cases}
\]
 and

\[
\textbf{(uHJB-FB)}\quad\begin{cases}
  &  \partial_t \tilde v(t,s,x) + \tilde U(x) f(t) + \sup_{\nu}\bigl \{ \tilde{\mathcal {H}}^{\tilde v}(t,s,x;\nu)\bigl \}=0,\; t<T\\
&\tilde v(T,s,x)=(1-F_\tau(T))U^P(x); (s,x)\in \mathbb R^m\times \mathbb R,
\end{cases}
\]

where
\begin{align*}
 \tilde{\mathcal {H}}^{\tilde v}(t,s,x;\nu)&=   \nu^\top b(t,s)\partial_x \tilde v(t,s,x)+\frac12\|\nu\sigma(t,s)\|^2\partial_{xx}\tilde v(t,s,x)\\
 &+\sum_{i=1}^m \partial_i \tilde v(t,s,x) b^i(t,s)s^i+\frac12 Tr(\Sigma(t,s)\Sigma(t,s)^\top D^2\tilde v(t,s,x))\\
 &+\sum_{i=1}^m \partial_{xs^i}\tilde v(t,s,x) \nu \sigma(t,s)(\sigma^i(t,s))^\top s^i.
\end{align*}
We then get the following theorem as a direct application of Ito's formula with localization procedure similar to the proof of Theorem \ref{thm:verif}.
\begin{theorem}
    Assume that Hypothesis B (resp. Hypothesis A) holds and there exists a solution of $\textbf{(bHJB-FB)}$ (resp.  $\textbf{(uHJB-FB)}$) denoted by $\tilde v$ with polynomial growth in $s,x$. Then, the optimal contract is given by 
   \[\xi^\star= \frac1{\varsigma+\eta} [\log(\frac{\eta\rho}\varsigma)+\varsigma X_{T\wedge \tau}],\]
    where $\hat \rho$ is the optimizer of 
    \[\inf_{\rho\geq 0} \big\{ \rho(\frac{\varsigma+\eta}{\varsigma})   (\frac{\varsigma}{\eta\rho})^{\frac\eta{\varsigma+\eta}}\tilde V_0^{FB} -\rho R\Big\},\]
    where $\tilde V_0^{FB}=\tilde v(0,S_0,x)$ and the optimal trading strategy is given by 
    \[\hat\pi_t=\nu^\star(t,S_t,X_t),\]
    where $\nu^\star(t,s,x)$ is an optimizer of $\tilde{\mathcal H}^{\tilde v}(t,s,x;\nu)$.
\end{theorem}

\begin{remark}
    Unlike moral hazard, the optimal contract is a linear function of $X_{T\wedge \tau}$ which is often observed in first-best scenario in contract theory, see for example \cite{mastrolia2018moral}. Note that the random default only differs 
\end{remark}

\section{Portfolio manager's problem and second-order BSDE}
\label{app:2bsdecontract}

In this appendix, we provide a rigorous mathematical framework and introduce the weak formulation of the portfolio optimization problem, extending \cite{cvitanic2017moral,cvitanic2018dynamic} to the case of an erratic horizon. It is now well established that the agent's problem under moral hazard in contract theory can be reduced to solving a second-order BSDE when the volatility of the asset is controlled, as first explained in \cite{cvitanic2018dynamic}.\footnote{Note that the very recent results of Chiusolo and Hubert \cite{chiusolo2024new} circumvent the need for second-order BSDEs by taking a detour through a first-best reformulation. However, since we now have access to a general theory of 2BSDEs with jumps and random horizons, we prefer to rely directly on the 2BSDE approach rather than to employ the technical shortcut of Chiusolo and Hubert.}
The results used in this section rely on several new developments in the theory of second-order BSDEs with jumps. Second-order BSDEs were initially introduced in \cite{cheridito2007second,soner2012wellposedness}, then reformulated in a tractable framework for stochastic control in \cite{possamai2018stochastic}, and later extended to random horizons in \cite{lin2020second,lin2022random,gennaro20252bsde}. More recently, the theory has been generalized to jumps and general semimartingales in \cite{denis2024second,possamai2018stochastic}, together with explicit links to stochastic control problems in \cite{mastrolia2025agency}. The work \cite{gennaro20252bsde} is particularly relevant for our analysis, as it directly connects 2BSDEs with jumps and \emph{erratic termination} to the associated stochastic control problem.
We begin by establishing a rigorous probabilistic setting under which the portfolio manager’s optimization problem is well-posed for any fixed compensation scheme $\xi$, under general integrability assumptions. We then prove the correspondence between the 2BSDE with erratic horizon and the value of the optimization problem. Throughout this appendix, we focus on the case $\epsilon = 0$, in order to stay consistent with the model of Cadenillas, Cvitanic and Zapatero \cite{cadenillas2007optimal}, thereby extending \cite[Section~6.3]{cvitanic2018dynamic} to the case of portfolio optimization with erratic termination.
For the sake of simplicity, and in line with a more realistic contractual structure, we assume in this appendix that the price process $S$ is not contractible for the investor (client, principal). Consequently, we take $X$ as the only contractible variable and thus as the main canonical process in the probabilistic weak formulation developed below.

\subsection{Probabilistic framework and weak formulation}

Let $\Omega := C([0,T],\mathbb{R})$ be the canonical space endowed with the 
canonical process $X$ and canonical filtration 
$\mathbb{F}=(\mathcal{F}_t)_{0\le t\le T}$.
Let $\tau:\Omega\to[0,T]$ be a finite $\mathbb{F}$--measurable stopping time.  
We denote by $\mathbb{G}=(\mathcal{G}_t)_{0\le t\le T}$ the progressive enlargement 
of $\mathbb{F}$ with $\tau$, defined by
\[
\mathcal{G}_t := \bigcap_{s>t} 
\big( \mathcal{F}_s \vee \sigma(\tau\wedge s) \big),
\qquad 0\le t\le T,
\]
which satisfies the usual conditions. A control process is a $\mathbb{G}$-progressively measurable process 
$\pi=(\pi_t)_{0\le t\le T}$ taking values in a set $A\subset\mathbb{R}$ and satisfying Definition \ref{admissible}. For such a control $\pi$, we say that a probability measure $\mathbb{P}$ on 
$(\Omega,\mathcal{G}_T)$ belongs to the set $\mathcal{P}(\pi)$ if there exists a 
$\mathbb{G}$-Brownian motion $W^{\mathbb{P}}$ under $\mathbb{P}$ such that
\[
X_t
=
x + \int_0^{t\wedge\tau} \pi_s \sigma_s\, dW_s^{\mathbb{P}}
  + \int_0^{t\wedge\tau} \pi_s b_s\, ds,
\qquad 0\le t\le T,\quad \mathbb{P}\text{-a.s.}
\]

The global set of admissible probability measures on the enlarged space is
\[
\mathcal{P}
   := \bigcup_{\pi\in\mathcal{A}} \mathcal{P}(\pi),
\]
where $\mathcal{A}$ denotes the class of all $\mathbb{G}$-progressively measurable admissible controls.  The problem of the portfolio manager is then given by 

\begin{equation}\label{pbAgent:rigor}  \tag{A} V^A_0(x,\xi)=\sup_{(\pi,\mathbb P^\pi)\in \mathcal A\times \mathcal P} J^A(\pi;x,\xi),\end{equation}
where
\[J^A(\pi;x,\xi):=\mathbb E^{\mathbb P^\pi}[ U^A(\xi)].\]
\begin{remark}
    In order to apply the result in \cite{gennaro20252bsde} we need to enforce the $\mathcal P-$density hypothesis, see \cite[Hypothesis 2.4 ($\mathcal P_0$-Density Hypothesis)]{gennaro20252bsde}. It is in particular satisfied as soon as $\tau$ is independent of $\mathbb F$ with a distribution invariant under any probability in $\mathcal P$.
\end{remark}

\subsection{The portfolio manager's optimization: 2BSDE with erratic termination}

We consider the dynamic version of the value function
Following the same lines as \cite[Proof of Theorem 4.2]{cvitanic2017moral} and assuming that the family of probabilities in $\mathcal P$ satisfies \cite[2.1.3. Conditioning and concatenation of probability measures and Assumption 2.1 (iii)-(iv)-(v)]{possamai2018stochastic}, we deduce that $V_0=Y_0$ where $(Y,Z,U,K)$ is a solution to the following second-order BSDE

\begin{equation}
\label{2bsde:app}Y_t=U_A(\xi)+\int_{t\wedge \tau}^{T\wedge \tau} F_s(Z^X_s,\hat\sigma^2_s)ds-\int_{t\wedge \tau}^{T\wedge \tau} Z^X_s dX_s -\int_{t\wedge \tau}^{T\wedge \tau} U_s dH_s+K_T-K_t ,\end{equation}
where for any $\Sigma>0$ 
\[F_s(z,\Sigma):=\sup_{\nu\in B_s(\Sigma)} \{z_x\nu \sigma_s \theta_s\}, \; B_s(\Sigma):=\{\nu\in  C,\; \|\nu\sigma_s\|^2=\Sigma\},\]
and where $\hat\sigma$ is the quadratic variation of the process $X$ under each $\mathbb P\in \mathcal P$ (see \cite{karandikar1995pathwise,cvitanic2018dynamic}). 
\begin{proposition}[\cite{gennaro20252bsde}]
   Assume that $U_A(\xi)\in \mathbb L^{p,k}$ for some $1\leq k\leq p$ and $\sigma,b$ are bounded.\footnote{We refer to \cite[Section 2.3]{gennaro20252bsde} for the exact integrability condition on $U_A(\xi)$. Note that it is satisfied for exponential utilities as soon as $\xi$ has exponential moments of any order greater than the risk aversion parameter for any probability $\mathbb P\in \mathcal P$. } there exists a unique solution $(Y,Z^X,U,K)$ to the 2BSDE with erratic horizon \eqref{2bsde:app}. 
\end{proposition}

Following the same steps as \cite[section 5.4]{cvitanic2018dynamic}, without loss of generality, there exists an $\mathbb G-$predictable process $\Gamma^X$ such that 
\[dK_t=H_s(Z^X_s,\Gamma^X_s)-F_s(Z^X_s,\hat\sigma^2_s) -\frac12\hat\sigma^2_s\Gamma^X_s,\]
where
\[H_s(z,g)=\sup_\Sigma\{\frac12\Sigma g+F_s(z,\Sigma)\},\; \Sigma\in \{\|\nu\sigma_s\|^2,\; \nu\in C\}.\] We thus get the decomposition
\begin{align*}
U_A(\xi)&= Y_0+\int_{0}^{T\wedge \tau} Z_s^X dX_s + \int_{0}^{T\wedge \tau} U_s dH_s-\int_0^{T\wedge \tau} (H_s(Z_s^X,\Gamma_s^X)-\frac12 \sigma_s^2\Gamma_s^X)ds.
\end{align*}
Applying finally Ito's formula with the function $U_A^{-1}(x)=-\log(-x)/\eta$ we get the decomposition $\xi=Y_T^{Y_0,Z^X,U,\Gamma^X}$ from Definition \ref{def:contract}.

\section{Extension to Merton's portfolio problem with consumption}\label{consumption extension}

 We propose to extend the study by assuming that the investor decides to consume part of the portfolio continuously over time with consumption strategy $c_t$ given by an $\mathbb F$-adapted random process. This problem was initially developed by Merton in \cite{merton1975optimum} and this extension will follow the same type of problem extended to random termination.  The value of the portfolio becomes
\[
    X_t := x + \int_0^t \pi_s \sigma_s dW_s + \int_0^t \pi_s b_s ds - \int_0^t c_s ds, \]
    or equivalently
\[ X_t:= x + \int_0^t \beta_s dW_s + \int_0^t (\beta_s \theta_s-c_s) ds.
\]
We denote by $\mathfrak C$ the set of admissible consumption strategies defined as the set of $\mathbb F-$adapted processes $c$ with values in $[0,\infty)$ such that \[\mathbb E[\int_0^T c_t^2dt]<\infty.\]
Definition \ref{def:contract} of admissible contract thus becomes the following:

\begin{definition}[Admissible contract with contractible variables and consumption] We denote by $\Xi$ the set of admissible contracts $\xi$ composed of $\mathcal G_{\tau\wedge T}-$measurable random variable $\xi=Y_{T\wedge \tau}^{Y_0,Z,Z^X,U,\Gamma^X,\Gamma}$, controlled by $\mathbb G-$predictable real-valued processes $U,Z=(Z^i)_{1\leq i\leq m},Z^X,\Gamma^X,\Gamma=(\Gamma_i)_{1\leq i\leq m}$ such that $\mathcal{I}_m - \Gamma^X_t \sigma_t \sigma^T_t$ is a positive definite matrix\footnote{$\mathcal{I}_m$ denotes the identity matrix in dimension $m$.} and
    \begin{align*}
   Y_t^{Y_0,Z,Z^X,U,\Gamma^X,\Gamma}&= Y_0+ \int_{0}^{t}\sum_{i=1}^m \frac{Z^i_r}{S^i_r} dS^i_r + \int_{0 }^{t} Z_r^X dX_r + \int_{0}^{t} U_r dH_r \\
    &+ \frac{1}{2}\int_{0}^{t} (\Gamma_r^X + \eta (Z_r^X)^2 ) d\langle X, X\rangle_r + \int_{0}^{t} \sum_{i=1}^m \frac{\Gamma_r^i}{S^i_t} d\langle S^i, X\rangle_r\\
    &- \int_{0}^{t} F(r, Z_r, Z_r^X, \Gamma_r, \Gamma_r^X, U_r,c_r) dr,
\end{align*}
for any $c\in\mathfrak C$ where \[ F(t, z, z_x, g, g_x, u,c)=\sup_{\nu\in C} f(t,z,z_x,g,g_x,u,\nu,c),\] with $f:[0,T]\times\Omega\times \mathbb R^m\times\mathbb R\times \mathbb R^m\times\mathbb R\times \mathbb R\times C\times [0,\infty)$ by
\begin{align*} f(t,z,z_x,g,g_x,u,\nu,c)&= z b_t +z_x(\nu \sigma_t \theta_t-c)+ \frac{1}{2} g_x\| \nu\sigma_t\|^2 - \frac{1}{2}\|\nu - \alpha_t\|^2 +\sum_{i=1}^m g^i \nu^i \sigma_t^i(\sigma_t^i)^\top\\
&- \frac{\lambda_t}{\eta} (\exp{(-\eta u) - 1})- \frac{\eta}{2}||z \sigma_t||^2,\end{align*}
 
and there exists $\eta'>\eta$ such that   \[\mathbb E\Big[\int_0^T (\|Z_s\|^2+\|Z^X_s\|^2+\|\Gamma_s\|+\|\Gamma^X_s\| +\|U_s\|^2\lambda_s) ds + \sup_{0\leq t\leq T} e^{\eta'|Y_t^{Y_0,Z,Z^X,U,\Gamma,\Gamma^X}|}\Big]<\infty.\]

   \end{definition}

We are assuming that the portfolio manager is receiving the contract $\xi$ and optimally chooses a strategy $\pi$ in order to stay close to a benchmark strategy $\alpha$ \textit{\`a la} Almgren-Chriss, see \cite{almgren2001optimal} so that the objective of the manager is to solve for any contract $\xi\in \Xi$ and consumption process $c\in \mathfrak C$ fixed by the investor
\begin{equation}\label{pbAgent-Appendix}  \tag{A} V^A_0(x,\xi,c)=\sup_{\pi\in \mathcal A} J^A(\pi;x,\xi,c),\end{equation}
where
\[J^A(\pi;x,\xi,c):=\mathbb E[ U^A(\xi- \int_{0}^{\T} \|\pi_s - \alpha_s\|^2 ds)].\]

The bilevel optimization corresponding to this portfolio problem with consumption becomes

\begin{equation}    \label{eq:principalPb:consumption}
    \tag{P}
     V_0^P(x):=\sup_{(\xi,\hat\pi,c)\in \mathcal C\times \mathcal A\times \mathfrak C}\mathbb E[U_P(X_{T\wedge \tau}-\xi)+\int_0^{T\wedge \tau} u_P(c_t)dt],
     \end{equation}
where $u_P:\mathbb R^+\longrightarrow\mathbb R$ is a utility function, non-convex and non-decreasing in the consumption control variable,

subject to
\begin{itemize}
    \item (R):\; $V_0^A(x,\xi,c)\geq R_0$
    \item (IC):\; $V_0^A(x,\xi,c) = J^A(\hat\pi;x,\xi,c)$.
\end{itemize}

Note that Theorem \ref{thm:value} is unchanged since the only impact of $c$ for the agent (portfolio manager) will be given through the compensation $Z^X$. This is consistent with Merton's approach where the optimal portfolio strategy $\hat\pi$ and consumption can be optimized independently in a first-best case model. We thus define the following system of coupled SDEs with jumps with solution $(X^{\pi^*,c},\hat Y)$ controlled by $(Z,Z^X,\Gamma,\Gamma^X,U,c)$

\[
(SDE)\begin{cases}
     &dX^{\pi^*,c}_t =  \pi_t^*\sigma(t,X_t) dW_s + (\pi_t^* b(t,S_t)-c_t) dt\\
     &d\hat Y_t= \sum_{i=1}^m \frac{Z^i_t}{S^i_t} dS^i_t +  Z_t^X dX^{\pi^*,c}_t +  U_t dH_t+ \frac{1}{2} (\Gamma_t^X + \eta (Z_t^X)^2 ) d\langle X^{\pi^*,c}, X^{\pi^*,c}\rangle_t\\
     &\qquad - F(t, Z_t, Z_t^X, \Gamma_t, \Gamma_t^X,U_t,c_t) dt+\sum_{i=1}^m \frac{\Gamma^i_t}{S^i_t}  d\langle S^i, X^{\pi^*,c}\rangle_t\\
     &X^{\pi^*,c}_0=x,\\
     &\hat Y_0=0. 
\end{cases}
\]

The problem of the investor becomes
\[
     V_0^P(x)=\sup_{(Z, Z^X, \Gamma, \Gamma^X, U,c)\in \mathcal U\times \mathfrak C}\;\mathbb E[U^P(X_{T\wedge \tau}-\hat Y_{T\wedge \tau})+\int_0^{T\wedge \tau} u_P(c_t)dt].\]
     Note that this class of problems does not exactly fit the framework \cite{blanchet2008optimal}. It requires to introduce a new state variable
     \[I_t:= \int_0^T u_P(c_t)dt\Longleftrightarrow  dI_t=u_P(c_t)dt,\\, I_0=0.\]
The problem of the investor thus becomes

\[
     V_0^P(x)=\sup_{(Z, Z^X, \Gamma, \Gamma^X, U,c)\in \mathcal U\times \mathfrak C}\;\mathbb E[U^P(X_{T\wedge \tau}-\hat Y_{T\wedge \tau})+I_{T\wedge \tau}].\]

\begin{itemize}
\item For bounded random time $\tau$ with support in $[0,T]$:

\begin{equation*}
  \hat V_0= \sup_{(Z, Z^X, \Gamma, \Gamma^X, U,c)\in \mathcal U\times \mathfrak C}\;  \E\bigl[ \int_{0}^T \big(U^P(X_t^{\pi^*,c} - \hat{Y}_t) +I_t\big) f(t)dt \bigl],
\end{equation*}
with corresponding HJB PDE

\begin{equation}\label{pde:bounded:consumption}
\begin{cases}
  \partial_t v(t,s,x,y,\iota) + (U^P(x - y)+\iota)f(t) + \sup_{(z^x, z, g^x, g, u,c)}  \mathcal H^v(t,s, x, y,\iota, \nabla v, \Delta v, z^x, z, g^x, g, u,c)=0,\; t<T\\
v(T,s,x,y,\iota)=0,\; (s,x,y,\iota)\in\mathbb R^m\times \mathbb R\times \mathbb R\times\mathbb R,
\end{cases}
\end{equation} where $\mathcal H^\phi$ is a differential operator given by
\begin{align*}
&\mathcal H^\phi(t,s,x,y,\iota,z^x,z,g^x,g,u,c):=\\
&\sum_{i=1}^m\phi_{s^i} s^ib^i(t,s)  + \phi_x \pi^*(z,z^x,g,g^x) \sigma(t,s) \theta(t,s)+ \mathcal L^\phi(t,s,x,y,\iota,c) \\
&+ \frac12 Tr(s \sigma(t, s) \sigma^T(t,s) s D^2\phi)+ \frac{1}{2}\phi_{xx}||\pi^*(z,z^x,g,g^x)  \sigma(t,s)||^2\\
& + \phi_y \Big[z b(t,s) + z^x \pi^*(z,z^x,g,g^x) \sigma(t,s) \theta(t,s)
+ \frac{1}{2}||\sigma(t,s) \pi^*(z,z^x,g,g^x)||^2(g^x + \eta |z^x|^2) \\
&-  F(t, z, z^x, g, g^x,u,c) + \pi^*(z,z^x,g,g^x) \sigma(t,s) \sigma(t,s)^T g\Big]
   \\
&+ \frac{1}{2}\phi_{yy} (||z_t \sigma(t,s)||^2 + z^x||\sigma(t,s) \pi^*(z,z^x,g,g^x)||^2) \\
&+ \phi_{xy} (z^x ||\pi^*(z,z^x,g,g^x) \sigma(t,s)||^2  + \pi^*(z,z^x,g,g^x) \sigma(t,s) (z \sigma)^T)\\
&+ \sum_{i=1}^m (\pi^*(z,z^x,g,g^x)  \sigma(t,s) \sigma^i(t,s)^T s^i)(\phi_{xs^i} + \phi_{ys^i}z^x) + \lambda_t(\phi(t, s,x, y+u) - \phi(t,s, x, y))\\
&+\sum_{i=1}^m \phi_{ys^i} z^i \sigma(t,s) \sigma^i(t,s)^T s^i,
\end{align*}
and with 
\[\mathcal L^\phi(t,s,x,y,\iota,c):=\phi_\iota u_P(c)-\phi_x c.\]

\begin{theorem}[Verification Theorem - bounded case]\label{thm:verif:consumption:bounded}
Assume that there exists a function $\phi$ twice continuously differentiable in space and differentiable in time, such that $\phi(t,s,x,y,\iota)$ solves \eqref{pde:bounded:consumption}. Furthermore, assume that $\phi$ has a quadratic growth in $y$ and polynomial growth in $s,x,\iota$ such that 
\[|\phi(t,s,x,y,\iota)|\leq \kappa(1+|x|^p+\|s\|^p+|\iota|^p+|y|^2),\; p>1,\; \kappa>0.\]

Then, for each $t \in [0,T]$, the strategy $(\hat{Z}^X, \hat{Z}, \hat{\Gamma}^X, \hat{\Gamma}, \hat{U},\hat c)$ is an optimal control for the investor's problem and

\begin{align*}
    \phi(0,S_0,x,0,0)=V_0 =  \sup_{(Z^X, Z, \Gamma^X, \Gamma, U,c)} \mathbb{E}\left[ U_P(X_{T \wedge \tau}^{\pi^*} - \hat{Y}_{T \wedge \tau})+I_{T\wedge \tau}\right]-\hat{Y_0}.
\end{align*}
Let $c^*(t,s,x,y,\iota)$ be an optimizer of $\mathcal L^\phi(t,s,x,y,\iota,c)$ such that (SDE) admits a unique solution with this choice of $c$. 
The optimal contract is given by 
\begin{align*}
    \xi^\star 
    &= \hat{Y}_0+ \int_{0}^{\T} \sum_{i=1}^m \frac{\hat{Z}^i_t}{S^i_t} dS^i_t + \int_0^{\T} \hat{Z}_t^X dX_t + \int_{0}^{\T} \hat{U}_s dH_s \\
    &+ \frac{1}{2}\int_0^{\T} (\hat{\Gamma}_t^X + \eta (\hat{Z}_t^X)^2 ) d\langle X, X\rangle_t - \int_0^{\T} F(t, \hat{Z}_t, \hat{Z}_t^X, \hat{\Gamma}_t, \hat{\Gamma}_t^X,\pi^*_t,\hat c_t) dt\\
    &+\int_0^{\T} \sum_{i=1}^m \frac{\hat{\Gamma}^i_s}{S^i_t} d\langle S^i, X\rangle_t,
\end{align*}
and the optimal consumption is given by $\hat c_t:=c^*(t,S_t,X^{\pi^*,\hat c}_t,\hat Y_t,I_t)$.
\end{theorem}

\item For unbounded random time with $[0,T]\subsetneq Supp(\tau)$
\begin{equation*}
  \hat V_0= \sup_{(Z, Z^X, \Gamma, \Gamma^X, U,c)\in \mathcal U\times \mathfrak C}\;  \E\bigl[ \int_{0}^T \big(U^P(X_t^{\pi^*,c} - \hat{Y}_t) +I_t\big) f(t)dt+(1-F_\tau(T))\big(U_P(X_T-\hat Y_T)+I_T\big) \bigl],
\end{equation*}
with corresponding HJB PDE is

\begin{equation}\label{pde:unbounded:consumption}
\begin{cases}
  \partial_t v(t,s,x,y,\iota) + (U_P(x - y)+\iota)f(t) + \sup_{(z^x, z, g^x, g, u,c)}  \mathcal H^v(t,s, x, y,\iota, \nabla v, \Delta v, z^x, z, g^x, g, u,c)=0,\; t<T\\
v(T,s,x,y,\iota)=(1-F_\tau(T))\big(U_P(x-y)+\iota\big),\; (s,x,y,\iota)\in\mathbb R^m\times \mathbb R\times \mathbb R\times\mathbb R,
\end{cases}
\end{equation}
with the same $\mathcal H$ operator as in the bounded case.

\begin{theorem}[Verification Theorem - unbounded case]\label{thm:verif:consumption:unbounded}
Assume that there exists a function $\phi$ twice continuously differentiable in space and differentiable in time, such that $\phi(t,s,x,y,\iota)$ solves \eqref{pde:unbounded:consumption}. Furthermore, assume that $\phi$ has a quadratic growth in $y$ and polynomial growth in $s,x,\iota$ such that 
\[|\phi(t,s,x,y,\iota)|\leq \kappa(1+|x|^p+\|s\|^p+|\iota|^p+|y|^2),\; p>1,\; \kappa>0.\]

Then, for each $t \in [0,T]$, the strategy $(\hat{Z}^X, \hat{Z}, \hat{\Gamma}^X, \hat{\Gamma}, \hat{U},\hat c)$ is an optimal strategy for the investor's control problem and

\begin{align*}
    \phi(0,S_0,x,0,0)=V_0 =  \sup_{(Z^X, Z, \Gamma^X, \Gamma, U,c)} \mathbb{E}\left[ U_P(X_{T \wedge \tau}^{\pi^*} - \hat{Y}_{T \wedge \tau})+I_{T\wedge \tau}\right]-\hat{Y_0}.
\end{align*}
Let $c^*(t,s,x,y,\iota)$ be an optimizer of $\mathcal L^\phi(t,s,x,y,\iota,c)$ such that (SDE) admits a unique solution with this choice of $c$. 
The optimal contract is given by 
\begin{align*}
    \xi^\star 
    &= \hat{Y}_0+ \int_{0}^{\T} \sum_{i=1}^m \frac{\hat{Z}^i_t}{S^i_t} dS^i_t + \int_0^{\T} \hat{Z}_t^X dX_t + \int_{0}^{\T} \hat{U}_s dH_s \\
    &+ \frac{1}{2}\int_0^{\T} (\hat{\Gamma}_t^X + \eta (\hat{Z}_t^X)^2 ) d\langle X, X\rangle_t - \int_0^{\T} F(t, \hat{Z}_t, \hat{Z}_t^X, \hat{\Gamma}_t, \hat{\Gamma}_t^X,\pi^*_t,\hat c_t) dt\\
    &+\int_0^{\T} \sum_{i=1}^m \frac{\hat{\Gamma}^i_s}{S^i_t} d\langle S^i, X\rangle_t,
\end{align*}
and the optimal consumption is given by $\hat c_t:=c^*(t,S_t,X^{\pi^*,\hat c}_t,\hat Y_t,I_t)$.
\end{theorem}

\end{itemize}
\begin{remark}
Note that, if we instead assume that the investor uses singular controls for the consumption strategy $c$, modeled as a non-decreasing process with bounded variation (see, for example, \cite{shreve1994optimal}), this would lead to
\[ X_t:= x + \int_0^t \beta_s dW_s + \int_0^t \beta_s \theta_s ds -\int_0^t dc_s.
\]

This case involves several difficulties, including 
\begin{enumerate}
\item second-order BSDE with jumps driven by general martingales to solve the portfolio manager problem.  
\item an obstacle PDE with gradient constraint for the investor's optimization with random horizon. 
\end{enumerate}
The first point is critical and requires to link the manager problem to a very new type of 2BSDE. Recent advances in this direction have been made in \cite{possamai2025mind}. However, the link with stochastic control problem with singular control remains an open problem beyond the scope of our paper that we will leave for future research. 
\end{remark}

\section{The agent optimal strategy with power utility}\label{AppendixPower}
Let us now consider, under the same setting of a Geometric Brownian Motion for the risky assets, the same portfolio evolution and therefore the same trading strategy $\pi$ as in Section \ref{chap3}. The agent's problem is the same, but for the sake of simplicity let us assume that there is no running penalty, so that the agent problem becomes
\begin{equation}\label{pbPowerAgent}  \tag{Power-A} V^A_0(x,\xi)=\sup_{\pi\in \mathcal A} J^A(\pi;x,\xi),\end{equation}
where
\[J^A(\pi;x,\xi):=\mathbb E[ U^A(\xi)].\]
with 
\[
    U^A(x) = \frac{1}{1-\eta}x^{1-\eta},\; 0\leq \eta<1
\]
We then have the following result.

\begin{theorem}
    Assume that \ref{densityhyp} and Hypothesis A or Hypothesis B are satisfied. For any $\xi\in \Xi$, let us further assume that there exists a $(C, Z, Z^X, \Gamma, \Gamma^X, U)$ solution to 
    \[
        dC_t = C_{t^-} dY_t
    \]
    with $C_{\T} = \xi$ and $Y$ a semi-martingale with the following decomposition
    \begin{align}\label{contract}
        dY_t&= Z_t(b_t dt + \sigma_t dW_t) + Z_t^X \pi_t(\sigma_t dW_t + b_t dt)\nonumber \\
        &- \frac{1}{2} (\Gamma_t^X -\eta (Z_t^X)^2 ) \pi_t \sigma_t \sigma_t^T \pi^T dt + \pi_t \sigma \sigma^T \Gamma_t dt\\
        &- F(t, Y_t, Z_t, Z_t^X, \Gamma_t, \Gamma_t^X, U_t) dt + U_t dH_t \nonumber
    \end{align}
    Then, the optimal strategy solving \eqref{pbPowerAgent} is 
    \begin{equation}
        \hat\pi_t =\pi^*(Z_t, Z_t^X, \Gamma_t, \Gamma^X_t),\text{ with }\pi^*(z,z_x,g,g_x):= proj(e_t, C)
    \end{equation}
    and the optimal value is given by $V^A_0(x, \xi) = \frac{1}{1-\eta}C_0^{1-\eta}$, where
    \[
        q_t = \frac{1}{\Gamma^X_t}(Z^X_t (\sigma_t \sigma_t^{\top})^{-1}b_t + \Gamma_t)
    \]
    and
    \[
        Q_t = \Gamma^X_t \sigma_t \sigma_t^{\top}
    \]
\end{theorem}

\begin{proof}
    The proof follows a similar idea to the one in the main body of the paper: we want to find a process related to our optimization problem and ensure that the driver of a related BSDE is such that the process is a super-martingale for each strategy and a martingale for the optimal strategy. In this case, in order to achieve this, we are further assuming that the auxiliary process solves the auxiliary exponential BSDE. As before, we define the family of processes indexed by the strategy $\pi$
    \begin{equation*}
        R^{\pi}_{t} = U^A(C_t).
    \end{equation*}
    This time we set $C = (C_t)_{t \in [0, T]}$ as the solution of the following simple BSDE, $dC_t = C_{t^-} dY_t$, coupled with $C_{\T} = \xi$, where the semi-martingale $Y$ follows the dynamics:
    \begin{align*}
       Y_t^{Y_0,Z,Z^X,U,\Gamma^X,\Gamma}&= Y_0+ \int_{0}^{t}\sum_{i=1}^m \frac{Z^i_r}{S^i_r} dS^i_r + \int_{0 }^{t} Z_r^X dX_r + \int_{0}^{t} U_r dH_r \\
        &- \frac{1}{2}\int_{0}^{t} (\Gamma_r^X -\eta (Z_r^X)^2 ) d\langle X, X\rangle_r + \int_{0}^{t} \sum_{i=1}^m \frac{\Gamma_r^i}{S^i_t} d\langle S^i, X\rangle_r\\
        &- \int_{0}^{t} F(r, Z_r, Z_r^X, \Gamma_r, \Gamma_r^X, U_r) dr,
    \end{align*}
    or equivalently the infinitesimal decomposition in \eqref{contract}.
 Therefore,
    \[
        R_t^{\pi} = \psi(C_t)
    \]
    where $\psi( x) = \frac{1}{1-\eta}x^{1-\eta}$. Using a simple Ito's expansion, we get that
    \begin{align*}
        d R_t &= \partial_t\psi dt + \partial_x\psi dC^c_t + \frac{1}{2}\partial_{xx}\psi d\langle C^c_t, C^c_t\rangle + (\psi(C_t) - \psi(C_{t^-}))dH_t\\
        &= (1-\eta) \frac{1}{1-\eta} C_t^{-\eta} C_{t} dY_{t} - \frac{1}{2}(1-\eta)\eta \frac{1}{1-\eta} C_t^{-\eta-1} C_{t}^2  d\langle Y^c_t, Y^c_t\rangle+ (\psi(C_t) - \psi(C_{t^-}))dH_t\\
    \end{align*}
    In this way, we observe that we can write the first two terms as $R_t$ times some other quantities. At the same time, since $C_t = C_{t-}(1 + U_t)$ when $H$ jumps,
    \[
    \psi(C_t) - \psi(C_{t-}) 
    = \frac{1}{1-\eta} \left( C_{t-}(1 + U_t) \right)^{1-\eta} - \frac{1}{1-\eta} C_{t-}^{1-\eta}
    = \frac{1}{1-\eta} C_{t-}^{1-\eta} \left( (1 + U_t)^{1-\eta} - 1 \right).
    \]
    Therefore, we can write the increment of $R$ at time $t$ as
    \[
    dR_t = (1-\eta) R_{t^-}\, dY_t^c
        - \frac{1}{2}(1-\eta)\eta R_{t^-}\, d\langle Y^c, Y^c\rangle_t
        + R_{t^-} \frac{(1 + U_t)^{1-\eta} - 1}{1-\eta}\, dH_t.
    \]
    We now expand \(dY_t\), whose continuous part satisfies:
    \[
    dY_t^{c}
    = Z_t b_t\,dt
    + Z_t \theta_t\,dW_t
    + Z_t^{X}\,\pi_t \sigma_t\,dW_t
    + Z_t^{X}\,\pi_t \sigma_t\,\theta_t\,dt
    - \frac{1}{2}\big(\Gamma_t^{X} -\eta(Z_t^{X})^{2}\big)\,
      \pi_t \sigma_t \sigma_t^{\top} \pi_t^{\top}\,dt
    - F_t\,dt
    + \pi_t \sigma_t \sigma_t^{\top} \Gamma_t\,dt .
    \]
    Its quadratic variation is
    \[
    d \langle Y^{c}, Y^{c} \rangle_t
    = Z_t \sigma_t \sigma_t^{\top} Z_t^{\top}\,dt
    + (Z_t^{x})^{2}\, \pi_t \sigma_t \sigma_t^{\top} \pi_t^{\top}\,dt .
    \]
    Substituting into the expression for $dR_t$, we get that
    \begin{align*}
        dR_t
        &= R_t \, (1-\eta) \Big(
            Z_t b_t
            + Z_t^{X} \pi_t \sigma_t \theta_t
            - \tfrac{1}{2} \Gamma_t^{X} 
              \pi_t \sigma_t \sigma_t^{\top} \pi_t^{\top}
            - F_t
            + \pi_t \sigma_t \sigma_t^{\top} \Gamma_t - \tfrac{\eta}{2} Z_t \sigma_t \sigma_t^{\top} Z_t
        \Big) dt\\
        &+ \, R_t \, (1-\eta) \big(Z_t \sigma_t + Z_t^{X} \pi_t \sigma_t \big)\, dW_t + R_t \big((1 + U_t)^{1 - \eta} - 1\big)\, dH_t .
    \end{align*}
    Now, because the terminal condition on $C$ is the reward for the agent, and the agent has a reserved utility, we know that the DDE of $C$ is positive and therefore $C_0 \geq 0$. Now, the utility is a power function, we have $R_0 \geq 0$ and
    \[
        dR_t = (1-\eta) R_{t^-} a_t dt +  R_{t^-} d\tilde{M}_t
    \]
    where $\tilde{M}$ is a $\G$-martingale and
    \[
        a_t =
            Z_t b_t
            + Z_t^{X} \pi_t \sigma_t \theta_t
            + \tfrac{1}{2} \Gamma_t^{X} 
              \pi_t \sigma_t \sigma_t^{\top} \pi_t^{\top}
            - F_t
            + \pi_t \sigma_t \sigma_t^{\top} \Gamma_t - \tfrac{\eta}{2} Z_t \sigma_t \sigma_t^{\top} Z_t
            + \frac{\lambda_t}{1-\eta} \big((1 + U_t)^{1-\eta} - 1\big) 
    \]
    Now, because of the positivity of the utility, we have that, in order for $R$ to be a super-martingale, we need to ensure that $a_t \leq 0$ and and optimality is achieved when the process is equal to $0$. Completing again the square, we have that
    \begin{align*}
        a_t = - \frac{1}{2}(\pi_t - q_t) Q_t (\pi_t - q_t)^{\top}
    \end{align*}
        where $Q_t = \Gamma_t^{X}\sigma_t \sigma_t^{\top}$
    and 
    \[
        q_t = \frac{1}{\Gamma^X_t}(Z^X_t (\sigma_t \sigma_t^{\top})^{-1}b_t + \Gamma_t)
    \]
    which is possible when $F$ has the following form:
    \[
        F = \frac12  \Big(q - \frac{\eta}{2\Gamma^X_t}Z_t^{\top} \Big)^{\top} Q \Big(q - \frac{\eta}{2\Gamma^X_t}Z_t^{\top} \Big) +  Z_t b_t + \frac{\lambda_t}{1-\eta}\Big[ (1 + U_t)^{1 - \eta} -1 \Big] + dist^2_Q(q_t, C)
    \]
    \end{proof}
\end{document}